\documentclass[letterpaper,11pt]{article}

\usepackage{verbatim}

\usepackage{newtxtext} %
\usepackage[margin=1in]{geometry}
\usepackage{setspace}
\usepackage{graphicx}
\usepackage{epsfig}

\usepackage{microtype} %
\usepackage{array} %

\usepackage{amssymb}
\usepackage{amsmath}
\usepackage{amsthm}
\usepackage{mathtools} %
\usepackage{stmaryrd} %
\usepackage{stackrel} %
\usepackage{latexsym} %

\usepackage{tikz}
\usetikzlibrary{arrows}
\usetikzlibrary{arrows.meta}
\usetikzlibrary{decorations.pathreplacing}
\usetikzlibrary{cd}
\usepackage[all]{xy} %

\usepackage{algorithm}
\usepackage[noend]{algpseudocode}
\algnewcommand{\LineComment}[1]{\Statex\hspace{\algorithmicindent}\(\triangleright\) #1}
\algnewcommand\algorithmicforeach{\textbf{for each}}
\algdef{S}[FOR]{ForEach}[1]{\algorithmicforeach\ #1\ \algorithmicdo}
\usepackage{etoolbox} %

\usepackage{caption}
\usepackage{subcaption}

\usepackage{makecell} %
\usepackage{booktabs} %
\usepackage{multirow} %
\usepackage{tablefootnote} %

\usepackage[hidelinks]{hyperref} %
\usepackage[symbol]{footmisc}
\usepackage{makeidx} %
\usepackage{url} %
\usepackage{color}
\usepackage{xcolor}
\usepackage{xspace} %
\usepackage[normalem]{ulem} %
\usepackage{soul} %
\usepackage{ifpdf} %
\usepackage{extarrows}
\usepackage{picins}
\usepackage{stackengine}    
\usepackage{lipsum}  

\makeatletter
\expandafter\patchcmd\csname\string\algorithmic\endcsname{\itemsep\z@}{\itemsep=0.25ex}{}{}
\makeatother

\newcounter{usesmallsep}
\ifnum\the\value{usesmallsep}=1

    \newlength{\myitemsep}
    \newlength{\mytopsep}
    \usepackage{enumitem}
    \setlist[itemize]{leftmargin=\parindent,parsep=\parskip,
      listparindent=\parindent,itemsep=\myitemsep,topsep=\myitemsep}
    \setlist[enumerate]{leftmargin=\parindent,parsep=\parskip,
      listparindent=\parindent,itemsep=\myitemsep,,topsep=\myitemsep}
    \setlist[description]{font=\bfseries,leftmargin=\parindent,parsep=\parskip,
      listparindent=\parindent,itemsep=\myitemsep,topsep=\myitemsep}
    
    \newlength{\mypartitlesep}
    \usepackage[explicit]{titlesec}
    \titlespacing{\paragraph}{0pt}{\mypartitlesep}{\mypartitlesep}
    
    \newlength{\mythmsep}
    \newtheoremstyle{mythmstyle}
      {\mythmsep} %
      {\mythmsep} %
      {\itshape} %
      {} %
      {\bfseries} %
      {.} %
      {.5em} %
      {} %
    
    \newtheoremstyle{mydefstyle}
      {\mythmsep} %
      {\mythmsep} %
      {} %
      {} %
      {\bfseries} %
      {.} %
      {.5em} %
      {} %
    
    \theoremstyle{mythmstyle}
        \newtheorem{theorem}{Theorem}
        \newtheorem{proposition}[theorem]{Proposition}
        \newtheorem{lemma}[theorem]{Lemma}
        \newtheorem{corollary}[theorem]{Corollary}
        \newtheorem{fact}[theorem]{Fact}
        \newtheorem*{fact*}{Fact}
    \theoremstyle{mydefstyle}
        \newtheorem{definition}{Definition}
        \newtheorem{problem}{Problem}
        \newtheorem{assumption}{Assumption}
        \newtheorem{remark}{Remark}
        \newtheorem*{remark*}{Remark}
        \newtheorem{algr}[algorithm]{Algorithm}

    \newenvironment{proof}
        {\vspace{-0.9em}\begin{proof}}
        {\end{proof}\vspace{-0.4em}}

\else

    \theoremstyle{plain}
        \newtheorem{theorem}{Theorem}
        \newtheorem{proposition}[theorem]{Proposition}

        \newtheorem{observation}[theorem]{Observation}
        
        \newtheorem*{algr*}{Algorithm}
    \theoremstyle{definition}
        \newtheorem{definition}[theorem]{Definition}

        \newtheorem{remark}[theorem]{Remark}
        \newtheorem*{remark*}{Remark}
        \newtheorem*{example*}{Example}
        
    \usepackage{enumitem}
    \setlist[itemize]{leftmargin=\parindent}
    \setlist[enumerate]{leftmargin=\parindent}
    \setlist[description]{font=\bfseries,leftmargin=\parindent}

\fi

\newcommand{\Hm}{\mathsf{H}}

\newcommand{\Real}{\mathbb{R}}

\newcommand{\Pers}{\mathsf{Pers}}

\renewcommand{\bar}[1]{\overline{#1}}

\newcommand{\lbarrowspace}{\;}
\newcommand{\lrarrowsp}[1]{\xleftrightarrow{\lbarrowspace#1\lbarrowspace}}
\let\leftrightarrowsp\lrarrowsp

\newcommand{\incto}{\hookrightarrow}
\newcommand{\inctosp}[1]{\xhookrightarrow{\lbarrowspace#1\lbarrowspace}}
\newcommand{\bakincto}{\hookleftarrow}
\newcommand{\bakinctosp}[1]{\xhookleftarrow{\lbarrowspace#1\lbarrowspace}}

\newcommand{\given}{\,|\,}
\newcommand{\Set}[1]{\{#1\}}

\newcommand{\add}[1]{\mskip-2.5mu\rotatebox[origin=c]{-45}{${\rightarrow}$}{#1}}
\newcommand{\del}[1]{\mskip-2.5mu\rotatebox[origin=c]{-45}{${\leftarrow}$}{#1}}
\newcommand{\addel}[1]{\mskip-2.5mu\rotatebox[origin=c]{-45}{${\leftrightarrow}$}{#1}}

\let\emptyset\varnothing

\let\union\cup

\newcommand{\Bcal}{\mathcal{B}}

\newcommand{\Dcal}{\mathcal{D}}
\newcommand{\Fcal}{\mathcal{F}}

\newcommand{\Ical}{\mathcal{I}}

\newcommand{\Rcal}{\mathcal{R}}

\newcommand{\Ucal}{\mathcal{U}}

\newcommand{\Zbb}{\mathbb{Z}}

\newcommand{\dG}{\delta}
\newcommand{\DG}{\Delta}

\newcommand{\LG}{\Lambda}

\newcommand{\sG}{\sigma}

\newcommand{\varsG}{{\xi}}
\newcommand{\SG}{\Sigma}
\newcommand{\tG}{\tau}
\newcommand{\thG}{\theta}

\newcommand{\red}{\color{red}}

\newcommand{\gray}{\color{gray}}

\newcommand{\Dim}{p}

\newcommand{\birth}{b}
\newcommand{\death}{d}
\newcommand{\filtcnt}{m}
\newcommand{\simpcnt}{n}

\newcommand{\ccleq}{\sqsubset}

\newcommand{\fsimp}[2]{\sigma_{#2}}

\newcommand{\dpc}{\Dcal}
\newcommand{\pset}{P}
\newcommand{\dpccnt}{s}
\newcommand{\posf}{D}

\newcommand{\Ud}{{\mathcal U}}
\newcommand{\ucplx}{L}
\newcommand{\usimp}{{\xi}}

\newcommand{\extsimp}{\chi}

\makeatletter
\newcommand*{\da@rightarrow}{\mathchar"0\hexnumber@\symAMSa 4B }
\newcommand*{\da@leftarrow}{\mathchar"0\hexnumber@\symAMSa 4C }
\newcommand*{\xdashrightarrow}[2][]{%
  \mathrel{%
    \mathpalette{\da@xarrow{#1}{#2}{}\da@rightarrow{\;}{}}{}%
  }%
}
\newcommand{\xdashleftarrow}[2][]{%
  \mathrel{%
    \mathpalette{\da@xarrow{#1}{#2}\da@leftarrow{}{}{\;}}{}%
  }%
}
\newcommand{\xdashleftrightarrow}[2][]{%
  \mathrel{%
    \mathpalette{\da@xarrow{#1}{#2}\da@leftarrow\da@rightarrow{}{}}{}%
  }%
}
\newcommand*{\da@xarrow}[7]{%
  \sbox0{$\ifx#7\scriptstyle\scriptscriptstyle\else\scriptstyle\fi#5#1#6\m@th$}%
  \sbox2{$\ifx#7\scriptstyle\scriptscriptstyle\else\scriptstyle\fi#5#2#6\m@th$}%
  \sbox4{$#7\dabar@\m@th$}%
  \dimen@=\wd0 %
  \ifdim\wd2 >\dimen@
    \dimen@=\wd2 %
  \fi
  \count@=2 %
  \def\da@bars{\dabar@\dabar@}%
  \@whiledim\count@\wd4<\dimen@\do{%
    \advance\count@\@ne
    \expandafter\def\expandafter\da@bars\expandafter{%
      \da@bars
      \dabar@ 
    }%
  }%
  \mathrel{#3}%
  \mathrel{%
    \mathop{\da@bars}\limits
    \ifx\\#1\\%
    \else
      _{\copy0}%
    \fi
    \ifx\\#2\\%
    \else
      ^{\copy2}%
    \fi
  }%
  \mathrel{#4}%
  \!\!
}
\makeatother

\newcounter{desccounter}

\newcommand{\cancel}[1]

\let\defemph\emph

\begin{document}

\title{Computing Zigzag Vineyard Efficiently Including Expansions and Contractions\thanks{This research is partially supported by NSF grants CCF 2049010 and CCF 2301360.}}

\author{Tamal K. Dey\thanks{Department of Computer Science, Purdue University. \texttt{tamaldey@purdue.edu}}
\and Tao Hou\thanks{School of Computing, DePaul University. \texttt{thou1@depaul.edu}}
}

\date{}

\maketitle
\thispagestyle{empty}

\begin{abstract}
Vines and vineyard connecting a stack of persistence diagrams
have been introduced in the non-zigzag setting by 
Cohen-Steiner et al.~\cite{cohen2006vines}.
We consider
computing these vines over changing filtrations for zigzag persistence
while incorporating two more operations: expansions and contractions
in addition to the transpositions considered in the non-zigzag setting. 
Although expansions and contractions can be implemented in quadratic time in the non-zigzag case 
by utilizing the linear-time transpositions, 
it is not obvious how they can be carried out under the zigzag framework
with the same complexity. 
While transpositions 
alone can be easily conducted in linear time using the recent {\sc FastZigzag} algorithm~\cite{dey2022fast},
expansions and contractions pose difficulty in 
breaking the barrier of  cubic complexity~\cite{dey2021updating}. 
Our main result is that, 
the half-way constructed
up-down filtration in the {\sc FastZigzag} algorithm indeed can be used to achieve linear 
time complexity 
for transpositions and
quadratic time complexity for expansions and contractions, matching the time complexity of all
corresponding operations in the non-zigzag case.

\cancel{
We observe that eight atomic operations are sufficient for changing
one zigzag filtration to another and provide update
algorithms for each of them. Six of these operations that have
some 
analogues to one or multiple \emph{transpositions} in the non-zigzag case can be executed as efficiently
as their non-zigzag counterparts. This approach
takes advantage of a recently discovered algorithm for
computing zigzag barcodes~\cite{dey2022fast} by converting a zigzag filtration to a
non-zigzag one and then connecting barcodes of the two with a bijection. 
The remaining two atomic operations do not have a strict analogue in the non-zigzag case.
For them, we propose algorithms based on explicit maintenance of representatives (homology cycles) which can be useful in their own rights for applications requiring explicit updates of representatives.
}
\end{abstract}

\newpage
\setcounter{page}{1}

\section{Introduction}
\label{sec:intro}

Computation of the persistence diagram (PD) also called the barcode from a given filtration has turned out to
be a central task in topological data analysis (TDA). Such a filtration usually represents
a nested sequence of sublevel sets of a function. In scenarios where the function changes, the filtration and hence the PD may also change. The authors in~\cite{cohen2006vines} provided
an efficient algorithm for updating the PD over an atomic operation which
\emph{transposes} two consecutive simplex additions in the filtration. 
Using this atomic operation repeatedly, one can connect a series of filtrations  derived from a time-varying function 
and obtain a stack of PD's called the
\emph{vineyard}. 
The authors~\cite{cohen2006vines} show that updating the PD due to the atomic transposition can be computed in $O(m)$ time if $m$ simplices constitute the filtration. 
Recently, zigzag persistence~\cite{carlsson2010zigzag,carlsson2009zigzag-realvalue,dey2022fast,dey2023revisiting,kerber2019barcodes,maria2014zigzag,maria2016computing} 
has drawn considerable attention
in the community of topological data analysis.
In this paper, 
we explore efficient updates of barcodes over such atomic operations 
for \emph{zigzag filtrations}
and include two additional types of operations~\cite{dey2021updating} which are necessary for converting
a zigzag filtration to any other: 
\emph{expansions} that enlarge the filtration 
and \emph{contractions} that shrink the filtration.
Motivating examples are given in
Appendix A of~\cite{dey2021updating}
and Section 3.1 of~\cite{tymochko2020using}
for supporting these operations efficiently.

Although the work~\cite{cohen2006vines} introducing vineyard 
does not consider expanding and contracting a filtration, 
the two operations can be naturally defined in the non-zigzag setting.
An expansion would involve inserting a simplex addition at a certain location
in the non-zigzag filtration and a contraction would be its reverse.
For implementing an expansion, one can
attach a new simplex addition to the end of the current non-zigzag filtration 
and then bring it to the intended position by a series of transpositions~\cite{cohen2006vines}.
For a contraction, one does the reverse. 
The cost of the update 
is then $O(m^2)$
because:
(i)~attaching a simplex addition to the filtration
involves a reduction which takes $O(m^2)$ time;
(ii) bringing it to correct position needs $O(m)$ transpositions each consuming
$O(m)$ time. 
The main finding of this paper is that, although zigzag filtrations 
introduce more complications  to the barcode computation, 
updating  zigzag barcodes can be done
with the \emph{same} complexity  as for the non-zigzag 
case over all the operations.

As recognized in~\cite{dey2021updating}, eight
atomic operations are
necessary for any zigzag filtration
to transform to any other,
including four implementing transpositions~\cite{carlsson2010zigzag,carlsson2009zigzag-realvalue,dey2023revisiting,maria2014zigzag,maria2016computing,oudot2015zigzag}
(henceforth called \emph{switches}) and 
another four implementing expansions~\cite{maria2014zigzag,maria2016computing} and contractions.
An immediate temptation would be to handle the updates
by maintaining a non-zigzag filtration 
obtained from straightening the  zigzag filtration 
using a recently proposed algorithm called {\sc FastZigzag}~\cite{dey2022fast}.
Although we can easily implement the switches
in linear time 
using this conversion~\cite{dey2022fast}
and the algorithm in~\cite{cohen2006vines}, 
implementing the expansions and contractions in quadratic time becomes difficult.
In fact, the authors of the arXiv paper~\cite{dey2021updating} adopt this approach 
but cannot  implement all operations due
to an adjacency change observed for certain operations.

To see the challenge
posed by the adjacency change,
recall that
the conversion to non-zigzag in~\cite{dey2022fast} first goes through an (implicit) conversion to 
an \emph{up-down} filtration~\cite{carlsson2009zigzag-realvalue},
which is then converted to a non-zigzag one. 
Since the first part of the up-down filtration
is the same as the first part of the final non-zigzag filtration, 
the adjacency change can be observed 
from the up-down filtration.
In the constructed up-down filtration,
the first part 
consists of all added simplices 
(with the same order) in the original zigzag filtration $\Fcal$.
Consequently, the first part may contain
multiple additions of the same simplex $\sigma$ without any deletion in between~\cite{dey2022fast}.
These multiple additions of $\sG$ are distinguished by \emph{$\Delta$-cells}~\cite{79945ceb-9a4c-3c90-a4ac-148c36996187,hatcher2002algebraic} 
which are copies of $\sG$ sharing
boundaries \emph{non-trivially} as in Figure~\ref{fig:bubble}.
Also, one needs to consider two different types of expansions and contractions,
where an expansion inserts two inclusion arrows involving the same simplex $\sG$
in a zigzag filtration, and a contraction as a reverse operation removes them
(see Section~\ref{sec:update-oper} for formal definitions).
Since the two inclusion arrows of $\sG$ being inserted or removed
could either point toward or away from each other,
we have \emph{inward} and \emph{outward} expansions
along with the corresponding \emph{inward} and \emph{outward} contractions.

The outward expansions and contractions turn out to be harder
than their inward counterparts.
To see this, consider
an outward expansion 
shown in  Figure~\ref{fig:out-expn-adj-chng},
where simplices added in the original zigzag filtrations
and corresponding $\DG$-cells added in the up-down filtrations
are connected by dotted blue lines.
Moreover, in the figure,
we connect certain $\DG$-cells  (resp.\ simplices)
and their faces
by blue arrows.
Observe that
before the expansion, the $\DG$-cells $\tG_0$, $\tG'_0$
have the $\DG$-cell $\sG_1$ 
in their boundaries.
However,
after the expansion, the boundary cell of $\tG_0$, $\tG'_0$
changes to $\sG_2$,
which is the $\DG$-cell corresponding to the newly inserted
addition of $\sG$ (in red).
Assuming working on non-zigzag filtrations,
such a change would mean that
one $\DG$-cell in
the boundary column of $\tG_0$ or $\tG'_0$
in the `persistence boundary matrix'~\cite{edelsbrunner2000topological}
changes into a \emph{much later} $\DG$-cell. Notice that there could be $O(m)$ such changes.
In~\cite{cohen2006vines}, a transposition swaps only two 
\emph{consecutive} columns and rows in the persistence boundary matrix~\cite{edelsbrunner2000topological} without
any change on the simplices' boundaries. 
The simplicity and locality of the transposition~\cite{cohen2006vines} 
make the linear algorithm possible. 
However, for the outward expansions and contractions,
it is not {at all} clear that the change in adjacency
arising in our case can be handled by a {series} of these transpositions~\cite{cohen2006vines}.
A straightforward approach 
would be to perform the $O(m^2)$ reductions~\cite{edelsbrunner2000topological} 
on each changed column
{in the persistence boundary matrix}. 
But then one may incur a cubic complexity to take care of the possible $O(m)$ such changes. %
Notice that
similar difficulty arises for outward contraction but not for inward expansion or contraction (see  Figure~\ref{fig:in-expn-adj-chng}).

\begin{figure}[!tb]
  \centering
  \includegraphics[width=0.95\linewidth]{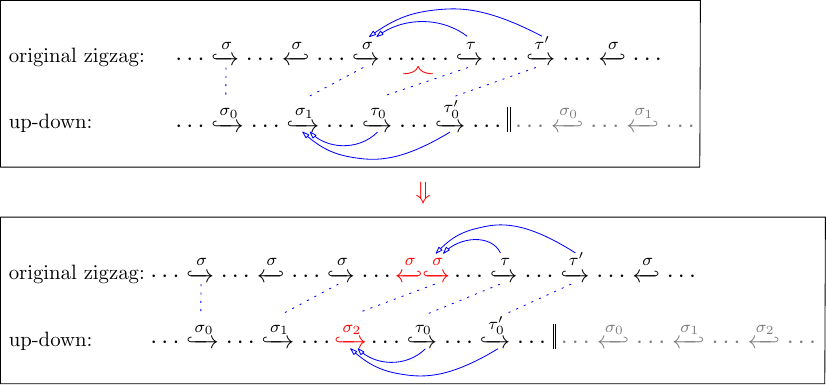}
  \caption{(upper box) $\Delta$-cells $\tau_0$ and $\tau_0'$ corresponding to simplices $\tau$ and $\tau'$  in the original filtration have $\Delta$-cell $\sigma_1$ (corresponding to an addition of $\sG$)
  in their boundaries in the up-down filtration; (lower box) $\Delta$-cells $\tau_0$ and $\tau_0'$ now have $\Delta$-cell $\sigma_2$  (corresponding to the newly inserted $\sigma$) in their boundaries
  once an outward expansion with $\sigma$ takes place.}
  \label{fig:out-expn-adj-chng}
\end{figure}

\begin{figure}[!tb]
  \centering
  \includegraphics[width=0.85\linewidth]{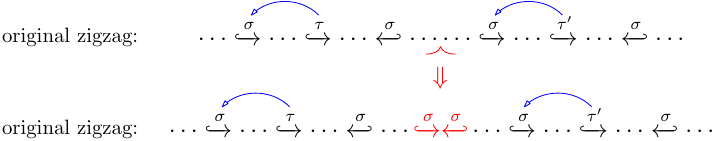}
  \caption{Each copy of simplex added in the original zigzag filtrations 
  has the same copies of simplices
  in its boundary before and after an inward expansion.
  So the boundary $\DG$-cells in the up-down filtrations also
  stay the same.}
  \label{fig:in-expn-adj-chng}
\end{figure}

In this paper, we find a way to tame down the curse of adjacency change
by adopting primarily two ideas. First,
to reduce the effect of the adjacency change  in the outward expansions and contractions,
we choose to work on up-down filtrations~\cite{carlsson2009zigzag-realvalue} instead.
Up-down filtrations eliminate the further complication on the adjacency introduced by the `coning'~\cite{cohen2009extending} of $\Delta$-cells
which is the final step of the conversion in~\cite{dey2022fast}.
Second, we introduce a collapsing strategy for identifying boundaries
which allows
us to nullify the effect of adjacency changes
on the time complexity.
To explain intuitively the strategy, consider an
outward contraction, where the adjacency change 
can be briefly stated as two $\Dim$-cells $\sG_1,\sG_2$ being
identified as the same cell.
Such a change can be
enacted by introducing an extra
$(\Dim+1)$-cell $\extsimp$ whose boundary equals $\sG_1+\sG_2$
and then collapsing $\extsimp$ (see Figure~\ref{fig:collapse}).
Section~\ref{sec:fzz-out-contra}  details the adjacency change and the algorithm.
Another key advantage of working on up-down filtrations 
is that
updates on up-down filtrations 
only need to maintain 
at most three 
chains per bar 
(see Remark~\ref{rmk:constant-bars}),
which is much less than 
the $O(\filtcnt)$ chains needed if working directly on 
the original zigzag filtrations~\cite{dey2021updating}.

\cancel{It turns out that
on the converted up-down filtrations, some updates can still be directly performed with existing algorithms.
For example,
using the algorithm by Maria and Oudot~\cite{maria2014zigzag,maria2016computing},
one can conduct the updates for the four switches
in linear time  on the up-down filtrations.
Moreover, one expansion
operation can also
be done in quadratic time using their algorithms~\cite{maria2014zigzag,maria2016computing}. 
The main challenge then boils down to
carrying out the
remaining  expansions and contractions
in quadratic time.}

\cancel{
{\red To illustrate the challenges,
we first distinguish two different types of expansions and contractions.
An expansion inserts two inclusion arrows involving the same simplex $\sG$
in a zigzag filtration, and a contraction as a reverse operation removes them.
Since the two inclusion arrows of $\sG$ being inserted 
could either point toward or away from each other,
we have \emph{inward} and \emph{outward} expansions
along with the corresponding \emph{inward} and \emph{outward} contractions
(see Figure~\ref{fig:expan-assoc-incs}).
The outward expansions and contractions turn out to be harder
than their inward counterparts.
The difficulty can be explained as follows.
Observe that the same simplex $\sG$ could be repeatedly added and deleted multiple times
in a zigzag filtration. 
Whenever $\sG$ is added in a zigzag filtration,
we always have an associated deletion of $\sG$ in the filtration
which is the nearest one after the current addition of $\sG$.
One property of inward expansion is that the two inserted
arrows are always associated addition and deletion of $\sG$,
while the other arrow associations are preserved
(see Figure~\ref{fig:in-assoc-incs} for an example).
In contrast, in an outward expansion,
the two inserted arrows of $\sG$
breaks an existing association and forces
a re-matching
(see Figure~\ref{fig:out-assoc-incs} for an example).
To illustrate the effect of such a re-matching,
we first recall that 
the up-down filtrations in~\cite{dey2022fast} 
are constructed with {$\DG$-complexes}~\cite{hatcher2002algebraic} 
made up of \emph{$\DG$-cells},
which are copies of simplices added
at different places in the original zigzag filtration
(see~\cite{dey2022fast} for details).
In outward expansion,
the re-matching of associated additions and deletions 
forces the adjacency of certain $\DG$-cells in the corresponding
up-down filtrations to change.
For example, $\DG$-cells corresponding to
$\tau$ and $\tau'$ in Figure~\ref{fig:out-assoc-incs} 
could share a common face
before the expansion
but could share no common faces afterward
(see Section~\ref{sec:fzz-out-contra} and~\ref{sec:fzz-out-expan}  for a detailed
explanation).
However, in inward expansion,
the adjacency of the $\DG$-cells in the 
up-down filtrations are preserved.
Similar effects also make  outward contractions harder than 
inward contractions.}}

Working on up-down filtrations, however, makes inward contraction non-trivial.
One may wonder the following: 
since inward contraction is the reverse
of inward expansion which already has an update algorithm~\cite{maria2014zigzag,maria2016computing},
can we not simply `reverse' this algorithm
to get an algorithm
for inward contraction?
We notice that this is usually not the case.
In fact, we discover an interval mapping 
for inward contraction different from that for inward expansion
presented in~\cite{maria2014zigzag,maria2016computing}, namely,
the persistence pairs 
are alternatively re-linked.
See Section~\ref{sec:fzz-in-contra}.

\section{Preliminaries}

\begin{figure}[t]
  \centering
  \captionsetup[subfigure]{justification=centering}

  \begin{subfigure}[t]{0.08\textwidth}
  \centering
  \includegraphics[width=\linewidth]{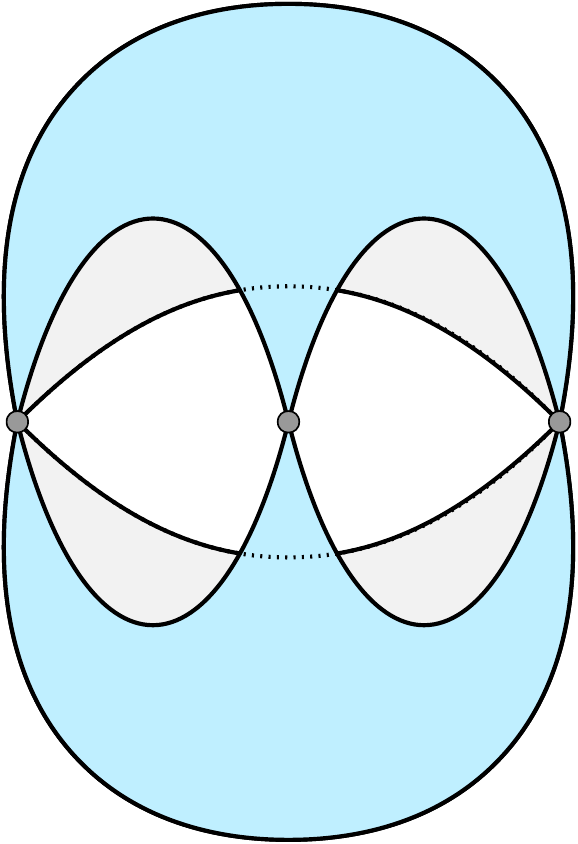}
  \label{fig:two-tri-0}
  \end{subfigure}
  \hspace{4em}
  \begin{subfigure}[t]{0.08\textwidth}
  \centering
  \includegraphics[width=\linewidth]{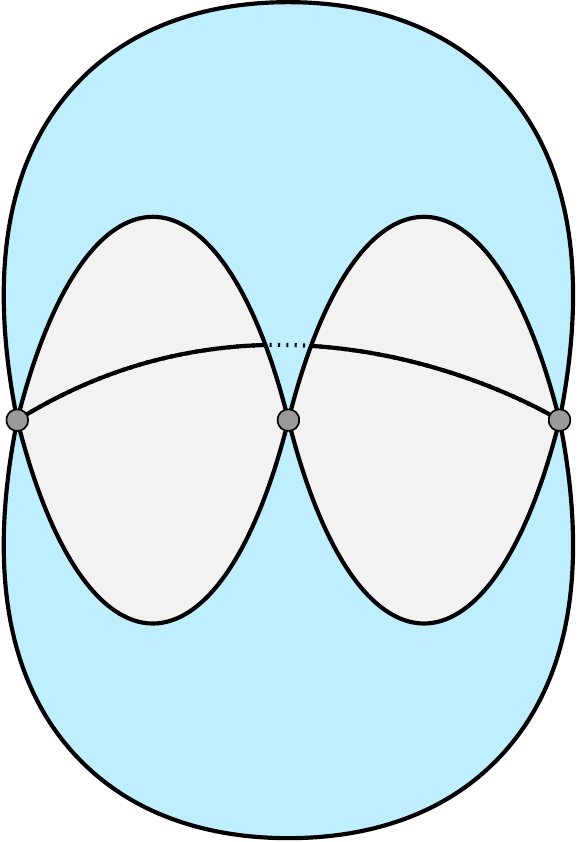}
  \label{fig:two-tri-1}
  \end{subfigure}
  \hspace{4em}
  \begin{subfigure}[t]{0.08\textwidth}
  \centering
  \includegraphics[width=\linewidth]{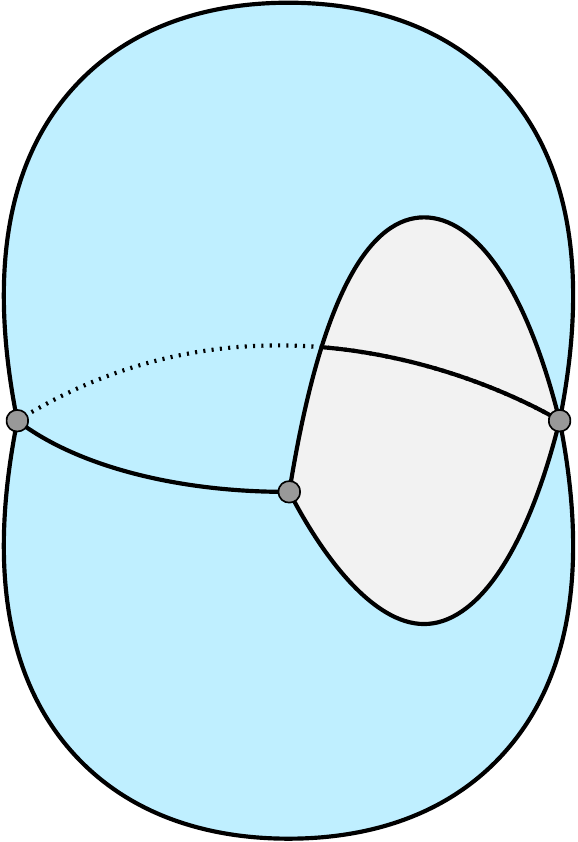}
  \label{fig:two-tri-2}
  \end{subfigure}
  \hspace{4em}
  \begin{subfigure}[t]{0.08\textwidth}
  \centering
  \includegraphics[width=\linewidth]{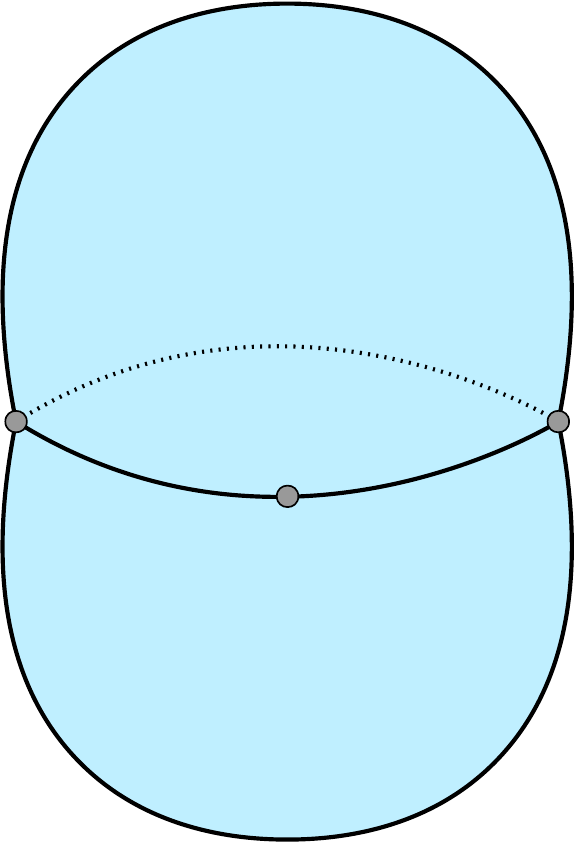}
  \label{fig:two-tri-3}
  \end{subfigure}

  \caption{Examples of $\DG$-complexes with two triangles \emph{having the same set of vertices} sharing 0, 1, 2, or 3 edges on their boundaries~\cite{dey2022fast}.}
  \label{fig:bubble}
\end{figure}

Since we utilize (partially) the conversion in~\cite{dey2022fast},
we work on \emph{$\DG$-complexes}~\cite{79945ceb-9a4c-3c90-a4ac-148c36996187,hatcher2002algebraic} 
instead of simplicial complexes
in this paper.
Building blocks of $\DG$-complexes are
called \emph{cells} or \emph{$\DG$-cells},
which are similar to simplices
but could have common faces in more
relaxed ways (see Figure~\ref{fig:bubble}
and also Definition 1 of \cite{dey2022fast} for a precise definition).
A {\it zigzag filtration} (or simply {\it filtration})
is  a sequence of $\DG$-complexes
\begin{equation}
\label{eqn:prelim-filt}
\Fcal: K_0 \leftrightarrow K_1 \leftrightarrow 
\cdots \leftrightarrow K_\filtcnt,
\end{equation}
in which each
$K_i\leftrightarrow K_{i+1}$ is either a forward inclusion $K_i\incto K_{i+1}$
(an addition of several cells)
or a backward inclusion $K_i\bakincto K_{i+1}$
(a deletion of several cells).
Taking the $\Dim$-th homology $\Hm_\Dim$,
we derive a {\it zigzag module}
\[
\Hm_\Dim(\Fcal): 
\Hm_\Dim(K_0) 
\leftrightarrow
\Hm_\Dim(K_1) 
\leftrightarrow
\cdots 
\leftrightarrow
\Hm_\Dim(K_\filtcnt). \]
In $\Hm_\Dim(\Fcal)$,
each $\Hm_\Dim(K_i)\leftrightarrow \Hm_\Dim(K_{i+1})$
is a linear map induced by inclusion.
In this paper, we take the coefficient $\Zbb_2$ for $\Hm_\Dim$
and hence  chains or cycles can be treated
as sets of cells.
The zigzag module $\Hm_\Dim(\Fcal)$
has a decomposition~\cite{carlsson2010zigzag,Gabriel72} of the form
$\Hm_\Dim(\Fcal)\simeq\bigoplus_{k\in\LG}\Ical^{[\birth_k,\death_k]}$,
in which each $\Ical^{[\birth_k,\death_k]}$
is an
{\it interval module} over the interval $[\birth_k,\death_k]\subseteq\Set{0,1,\ldots,\filtcnt}$.
The (multi-)set of intervals
$\Pers_\Dim(\Fcal):=\Set{[\birth_k,\death_k]\given k\in\LG}$
is an invariant of $\Hm_\Dim(\Fcal)$
and is called the $\Dim$-th {\it zigzag barcode} (or simply {\it barcode}) of $\Fcal$.
Each interval in $\Pers_\Dim(\Fcal)$
is called a $\Dim$-th \emph{persistence interval}.
We usually consider the homology $\Hm_*$ in all dimensions
and take the zigzag module $\Hm_*(\Fcal)$,
for which we have $\Pers_*(\Fcal)=\bigsqcup_{\Dim\geq 0}\Pers_\Dim(\Fcal)$.
In this paper,
sometimes a filtration may have nonconsecutive indices
on the complexes (i.e., some indices are skipped);
notice that the barcode is still well-defined.

An inclusion in a filtration is called \emph{cell-wise}
if it is an addition or deletion of a single cell $\sG$,
which we sometimes denote as, e.g., $K_i\leftrightarrowsp{\sG}K_{i+1}$.
A filtration is called \emph{cell-wise} if it contains only cell-wise inclusions.
For computational purposes, 
filtrations in this paper are by default
cell-wise filtrations
starting and ending with \emph{empty} complexes
(as adopted in~\cite{dey2022fast,maria2014zigzag});
notice that any filtration can be converted into this  standard form
by expanding the inclusions and attaching complexes to both ends.

\subsection{Update operations}
\label{sec:update-oper}

We now present all the update operations.
In this subsection, $\Fcal$ and $\Fcal'$  are both 
simplex-wise
zigzag filtrations consisting of simplicial complexes
which start and end with empty complexes.

\medskip
\noindent
\textit{{Forward switch}}~\cite{maria2014zigzag}
swaps two forward inclusions and
requires $\sG\nsubseteq\tG$:

\vspace{-0.5em}
\begin{equation}
\label{eqn:fwd-sw}
\begin{tikzpicture}[baseline=(current  bounding  box.center)]
\tikzstyle{every node}=[minimum width=24em]
\node (a) at (0,0) {$\Fcal:K_0 \leftrightarrow\cdots\leftrightarrow K_{i-1}\inctosp{\sG}K_i\inctosp{\tG}K_{i+1}\leftrightarrow\cdots\leftrightarrow K_\filtcnt$}; 
\node (b) at (0,-0.6){$\Fcal':K_0\leftrightarrow\cdots\leftrightarrow K_{i-1}\inctosp{\tG} K'_i\inctosp{\sG} K_{i+1}\leftrightarrow\cdots\leftrightarrow K_\filtcnt$};
\path[->] (a.0) edge [bend left=90,looseness=1.5,arrows={-latex},dashed] (b.0);
\end{tikzpicture}
\end{equation}

\noindent
Notice that if $\sG\subseteq\tG$, then adding $\tG$ to $K_{i-1}$ in $\Fcal'$
does not produce a simplicial complex.

\medskip
\noindent
\textit{{Backward switch}}
swaps two backward inclusions and
requires $\tG\nsubseteq\sG$:

\vspace{-0.5em}
\begin{equation}
\label{eqn:bak-sw}
\begin{tikzpicture}[baseline=(current  bounding  box.center)]
\tikzstyle{every node}=[minimum width=24em]
\node (a) at (0,0) {$\Fcal: K_0 \leftrightarrow
\cdots
\leftrightarrow 
K_{i-1}\bakinctosp{\sG} 
K_i 
\bakinctosp{\tG} K_{i+1}
\leftrightarrow
\cdots \leftrightarrow K_\filtcnt$}; 
\node (b) at (0,-0.6){$\Fcal': K_0 \leftrightarrow
\cdots
\leftrightarrow 
K_{i-1}\bakinctosp{\tG} 
K'_i 
\bakinctosp{\sG} K_{i+1}
\leftrightarrow
\cdots \leftrightarrow K_\filtcnt$};
\path[->] (a.0) edge [bend left=90,looseness=1.5,arrows={-latex},dashed] (b.0);
\end{tikzpicture}
\end{equation}

\noindent
\textit{{Outward/inward switch}}~\cite{carlsson2010zigzag}
swaps two inclusions of opposite directions
and requires $\sG\neq\tG$:

\vspace{-0.5em}
\begin{equation}
\label{eqn:out-in-sw}
\begin{tikzpicture}[baseline=(current  bounding  box.center)]
\tikzstyle{every node}=[minimum width=24em]
\node (a) at (0,0) {$\Fcal: K_0 \leftrightarrow
\cdots
\leftrightarrow 
K_{i-1}\inctosp{\sG} 
K_i 
\bakinctosp{\tG} K_{i+1}
\leftrightarrow
\cdots \leftrightarrow K_\filtcnt$}; 
\node (b) at (0,-0.6){$\Fcal': K_0 \leftrightarrow
\cdots
\leftrightarrow 
K_{i-1}\bakinctosp{\tG} 
K'_i 
\inctosp{\sG} K_{i+1}
\leftrightarrow
\cdots \leftrightarrow K_\filtcnt$};
\path[->] (a.0) edge [bend left=90,looseness=1.5,arrows={latex-latex},dashed] (b.0);
\end{tikzpicture}
\end{equation}

\noindent
The switch from $\Fcal$ to $\Fcal'$ is 
an {{outward}} 
switch
and the switch from $\Fcal'$ to $\Fcal$ is  an {{inward}} switch.
Notice that if $\sG=\tG$, then e.g., for outward switch, we cannot delete $\tG$ from $K_{i-1}$ in $\Fcal'$
because $\tG\not\in K_{i-1}$.

\medskip
\noindent
\textit{{Inward contraction/expansion}}~\cite{maria2014zigzag}
is as follows:

\vspace{-0.5em}
\begin{equation}
\label{eqn:in-contrac-expan}
\begin{tikzpicture}[baseline=(current  bounding  box.center)]
\node (a) at (0,0) {$\Fcal: K_0 \leftrightarrow
\cdots\leftrightarrow K_{i-2}
\leftrightarrow 
K_{i-1}\inctosp{\sG} 
K_i 
\bakinctosp{\sG} K_{i+1}
\leftrightarrow K_{i+2}\leftrightarrow
\cdots \leftrightarrow K_\filtcnt$}; 
\node (b) at (0,-0.9){$\Fcal': K_0 \leftrightarrow
\cdots\leftrightarrow K_{i-2}
\leftrightarrow 
K'_{i}
\leftrightarrow K_{i+2}\leftrightarrow
\cdots \leftrightarrow K_\filtcnt$};
\draw[->,arrows={latex-latex},dashed] (a.0) .. controls (+7,-0.1) and (+6.5,-0.8) .. (b.0);
\end{tikzpicture}
\end{equation}
\noindent
From $\Fcal$ to $\Fcal'$ we have an inward contraction 
and from $\Fcal'$ to $\Fcal$ we have an inward expansion.
To clearly show 
the correspondence of complexes in $\Fcal$ and $\Fcal'$, 
indices for $\Fcal'$ are made non-consecutive
in which $i-1$ and $i+1$ are skipped.
We also have $K_{i-1}=K'_{i}=K_{i+1}$.

\medskip
\noindent
\textit{{Outward contraction/expansion}}
is similar to the inward version with the difference
that the two center arrows now pointing away from each other:

\vspace{-0.5em}
\begin{equation}
\label{eqn:out-contrac-expan}
\begin{tikzpicture}[baseline=(current  bounding  box.center)]
\node (a) at (0,0) {$\Fcal: K_0 \leftrightarrow
\cdots\leftrightarrow K_{i-2}
\leftrightarrow 
K_{i-1}\bakinctosp{\sG} 
K_i 
\inctosp{\sG} K_{i+1}
\leftrightarrow K_{i+2}\leftrightarrow
\cdots \leftrightarrow K_\filtcnt$}; 
\node (b) at (0,-0.9){$\Fcal': K_0 \leftrightarrow
\cdots\leftrightarrow K_{i-2}
\leftrightarrow 
K'_{i}
\leftrightarrow K_{i+2}\leftrightarrow
\cdots \leftrightarrow K_\filtcnt$};
\draw[->,arrows={latex-latex},dashed] (a.0) .. controls (+7,-0.1) and (+6.5,-0.8) .. (b.0);
\end{tikzpicture}
\end{equation}

\noindent Notice that we also have $K_{i-1}=K'_{i}=K_{i+1}$.

\cancel{
\paragraph{Universality of the operations.}
We present the following fact:

\begin{proposition}\label{prop:universality}
Let $\Fcal_1,\Fcal_2$ be any two simplex-wise zigzag filtrations
starting and ending with empty complexes.
Then $\Fcal_1$ can be transformed into $\Fcal_2$ by
a sequence of the update operations listed above.
\end{proposition}
\begin{proof}
Based on~\cite{maria2014zigzag},
an empty filtration can always be transformed into $\Fcal_2$ 
by  backward switches and inward expansions.
Reversely, $\Fcal_1$ can be transformed into 
an empty filtration by backward switches and inward contractions.
We then have
the transformation from $\Fcal_1$ to $\Fcal_2$.
\end{proof}
}

\section{Up-down conversion and easier updates}
In this section,
we  first briefly overview the conversion to up-down filtrations in~\cite{dey2022fast}
 and provide some definitions for up-down filtrations.
We then describe how the conversion facilitates the updates
for certain operations based on existing algorithms~\cite{dey2023revisiting,maria2014zigzag}.
For operations whose updates are more involved,
we provide the algorithms in Section~\ref{sec:fzz-out-contra}--\ref{sec:fzz-in-contra}.

\subsection{Conversion to up-down filtrations}
\label{sec:fzz-up}

Given a \emph{simplex}-wise zigzag filtration 
\begin{equation}
\label{eqn:fzz-up-sec-filt}
\Fcal:
\emptyset=
K_0\leftrightarrowsp{\fsimp{}{0}} K_1\leftrightarrowsp{\fsimp{}{1}}
\cdots 
\leftrightarrowsp{\fsimp{}{\filtcnt-1}} K_\filtcnt
=\emptyset
\end{equation}
consisting of simplicial complexes,
we convert $\Fcal$ into the following \textit{cell}-wise 
\emph{up-down} filtration
consisting of $\DG$-complexes~\cite{dey2022fast}:
\begin{equation*}
\Ud:\emptyset=\ucplx_0\inctosp{\usimp_0}
\ucplx_1\inctosp{\usimp_{1}}
\cdots\inctosp{\usimp_{\simpcnt-1}} 
\ucplx_{\simpcnt}
\bakinctosp{\usimp_{\simpcnt}}
\ucplx_{\simpcnt+1}
\bakinctosp{\usimp_{\simpcnt+1}}
\cdots
\bakinctosp{\usimp_{\filtcnt-1}} \ucplx_{\filtcnt}=\emptyset.
\end{equation*}
Notice that an up-down filtration is a special type of 
zigzag filtration where the first half consists of only additions
and the second half consists of only deletions~\cite{carlsson2009zigzag-realvalue}.
In $\Ucal$,
cells $\usimp_0,\usimp_1,\ldots,\usimp_{\simpcnt-1}$ 
are 
distinct copies of simplices added in $\Fcal$ 
with the addition order preserved
(notice that the same simplex could be 
added  more than once in $\Fcal$).
Symmetrically,
cells 
$\usimp_{\simpcnt},\usimp_{\simpcnt+1},\ldots,\usimp_{\filtcnt-1}$ 
are 
distinct
copies of simplices deleted in $\Fcal$,
with the order also preserved.
To build $\Ucal$ 
from ${\Fcal}$, 
one only needs to 
list all the additions first and then the deletions in ${\Fcal}$,
following their orders in ${\Fcal}$,
with each addition of a simplex in $\Fcal$ corresponding to a distinct $\DG$-cell in $\Ucal$~\cite{dey2022fast}.
We  have that the conversion from $\Fcal$ to $\Ucal$
can be done in $O(\filtcnt)$ time~\cite{dey2022fast}.
We also have $\filtcnt=2\simpcnt$
because each added simplex
must be eventually deleted in ${\Fcal}$.
Figure~\ref{fig:fzz-out-contrac}
provides an example of the conversion, where an edge 
is added twice in $\Fcal$ with $\sG_1$ corresponding to its first addition
and $\sG_2$ corresponding to its second addition.
In the up-down filtration $\Ucal$,
$\sG_1$ and $\sG_2$ appear in the same complex forming parallel edges 
(also called multi-edges).

We then briefly describe how the intervals in $\Pers_*(\Fcal)$
and $\Pers_*(\Ucal)$ correspond.
First notice that each interval in $\Pers_*(\Fcal)$
or $\Pers_*(\Ucal)$ can be considered as `generated' by a pair
of simplex-wise or cell-wise inclusions
(i.e., simplex/cell additions or deletions).
For example, an interval $[b,d]\in\Pers_*(\Fcal)$ is generated
by 
the  inclusion $K_{b-1}\leftrightarrow K_{b}$ on the left and 
the  inclusion $K_{d}\leftrightarrow K_{d+1}$ on the right.
As indicated,
there is a natural bijection $\phi$
from the cell-wise inclusions in $\Ud$ to the simplex-wise inclusions in $\Fcal$.
For ease of presentation,
we assign each simplex-wise inclusion
$K_i\leftrightarrowsp{\fsimp{}{i}}K_{i+1}$ in $\Fcal$
the unique  \emph{index} $i$.
We also similarly index cell-wise inclusions in $\Ucal$.
We then let
the domain and codomain of $\phi$
be the sets of {indices} for the simplex-wise or cell-wise
inclusions.
We  have the following fact from~\cite{dey2022fast}:
\begin{theorem}\label{thm:fzz-intv-map}
There is a bijection from  $\Pers_*(\Ucal)$ to $\Pers_*(\Fcal)$
s.t.\ corresponding intervals 
are generated by corresponding pairs of inclusions.
Specifically, suppose that an interval in $\Pers_*(\Ucal)$
is generated by two cell-wise inclusions indexed at $l$
and $r$  ($l< r$).
Then its corresponding interval in $\Pers_*(\Fcal)$
is generated by  two
simplex-wise inclusions indexed at $\phi(l)$
and $\phi(r)$
(notice that  $\phi(r)$ may be less than $\phi(l)$).
Therefore,
in order to map an interval in $\Pers_*(\Ucal)$
back to the interval in $\Pers_*(\Fcal)$,
one only needs to look up the map $\phi$ which takes constant time.
\end{theorem}

\cancel{\gray
We then have the following fact from~\cite{dey2022fast}:
\begin{theorem}\label{thm:fzz-intv-map}
Based on the conversion from $\Fcal$ to $\Ucal$,
there is  a natural bijection $\phi$
from the additions and deletions in $\Ud$ to the additions and deletions in $\Fcal$.
Furthermore, 
there is  a bijection from
$\Pers_*(\Ucal)$ to $\Pers_*(\Fcal)$ s.t.\ given the bijection $\phi$, 
each interval in $\Pers_*(\Ucal)$ 
can be converted to its corresponding interval in $\Pers_*(\Fcal)$ 
in $O(1)$ time.
\end{theorem}

\begin{definition}\label{dfn:creator-destroyer}
In $\Fcal$ or $\Ud$, 
let each addition or deletion 
be uniquely identified by its index in the filtration,
e.g., the index of $K_i\leftrightarrowsp{\fsimp{}{i}}K_{i+1}$ in $\Fcal$ is $i$.
Then, 
the \defemph{creator} of an interval $[b,d]\in\Pers_*(\Fcal)\text{ or }\Pers_*(\Ud)$
is an addition/deletion indexed at $b-1$,
and the \defemph{destroyer} of $[b,d]$
is an addition/deletion indexed at $d$.
\end{definition}

As stated previously, 
each $\usimp_i$ in $\Ud$ for $0\leq i<\simpcnt$ corresponds
to an addition in $\Fcal$, and
each $\usimp_i$ for $\simpcnt\leq i<\filtcnt$ corresponds
to a deletion in $\Fcal$.
This naturally defines a bijection $\phi$
from the additions and deletions in $\Ud$ to the additions and deletions in $\Fcal$.
Moreover, for simplicity, we let the domain and codomain of $\phi$
be the sets of indices for the additions and deletions.
The interval mapping in~\cite{dey2022fast} (which uses the \emph{Mayer-Vietoris Diamond}~\cite{carlsson2010zigzag,carlsson2009zigzag-realvalue}) can be summarized
as follows:

\begin{theorem}
\label{thm:fzz-intv-map}
Given $\Pers_*(\Ud)$, one can retrieve $\Pers_*(\Fcal)$ 
using the following bijection
from $\Pers_*(\Ud)$
to $\Pers_*(\Fcal)${\rm:}
an 
interval $[b,d]\in\Pers_p(\Ud)$ 
with a creator indexed at $b-1$ and a destroyer indexed at $d$
is mapped to an interval $I\in \Pers_*(\Fcal)$ with
the same creator and destroyer indexed at 
$\phi(b-1)$ and $\phi(d)$ respectively.
Specifically, 
\begin{itemize}
    \item 
If $\phi(b-1)<\phi(d)$, then $I=[\phi(b-1)+1,\phi(d)]\in \Pers_p(\Fcal)$,
where $\phi(b-1)$ indexes the creator and $\phi(d)$ indexes the destroyer.
\item
Otherwise,
$I=[\phi(d)+1,\phi(b-1)]\in \Pers_{p-1}(\Fcal)$,
where $\phi(d)$ indexes the creator and $\phi(b-1)$ indexes the destroyer.
\end{itemize}
\end{theorem}

Notice the decrease in the dimension of the mapped interval in $\Pers_{*}(\Fcal)$
when $\phi(d)<\phi(b-1)$
(indicating a swap on the roles of the creator and destroyer).}

\subsection{Representatives for up-down filtrations}
Consider
an up-down filtration $\Ucal$  as built by the conversion in Section~\ref{sec:fzz-up}.
Since each cell in $\Ucal$ is added exactly once,
we can uniquely denote each  addition (resp.\ deletion)
of a cell $\varsG$ in $\Ucal$
as $\add{\varsG}$ (resp.\ $\del{\varsG}$).
We also use $\addel{\varsG}$ to denote
either the addition or deletion of $\varsG$.
As mentioned, each interval $[b,d]\in\Pers_*(\Ucal)$ 
is generated 
by a pair $(\addel{\varsG}, \addel{\varsG'})$,
where the locations of $\addel{\varsG}$ and $\addel{\varsG'}$ are as follows:
\begin{equation*}
\Ucal:
\cdots
\leftrightarrow
L_{b-1}
\lrarrowsp{\varsG}
L_b
\leftrightarrow
\cdots
\leftrightarrow
L_{d}
\lrarrowsp{\varsG'}
L_{d+1}
\leftrightarrow
\cdots
\end{equation*}
In such a pair,
we  call $\addel{\varsG}$ 
\emph{positive}
and $\addel{\varsG'}$ 
\emph{negative}.
From now on, for an $\Ucal$, 
we  interchangeably consider $\Pers_*(\Ucal)$ as all
 pairs of additions and deletions generating the intervals.
We also call a pair from $\Pers_\Dim(\Ucal)$ 
a \emph{$\Dim$-pair}.
\begin{definition}[Creators of chains]
Let $\add{\varsG}$ be an addition  in $\Ucal$.
A chain $A$ is said to be \defemph{created by} 
$\add{\varsG}$ if $\varsG\in A$ and all other cells in $A$
are added before $\varsG$ in $\Ucal$.
Moreover, Let $\del{\varsG}$ be a deletion  in $\Ucal$.
A chain $A$ is said to be \defemph{created by} 
$\del{\varsG}$ if $\varsG\in A$ and all other cells in $A$
are deleted after $\varsG$ in $\Ucal$.
We also say that $\add{\varsG}$ or $\del{\varsG}$ is the \defemph{creator} of $A$.
When the direction of the arrow for $\varsG$ is clear in the context,
we sometimes drop the arrow and simply
say that $A$ is created by $\varsG$, or  $\varsG$ is the creator of $A$.
\end{definition}
\begin{definition}\label{dfn:ud-rep}
We classify
the pairs in $\Pers_*(\Ucal)$
and define their \defemph{representatives} respectively as follows:
\begin{itemize}
    \item A $\Dim$-pair of the form $(\add{\gamma},\add{\eta})$ is called \defemph{closed-open}.
    The representative for $(\add{\gamma},\add{\eta})$
    is a tuple of chains $(z,A)$ with $z=\partial(A)$,
    where $z$ is a $\Dim$-chain 
    created by $\add{\gamma}$
    and $A$ is a $(\Dim+1)$-chain
    created by $\add{\eta}$.

    \item A $\Dim$-pair of the form $(\del{\gamma},\del{\eta})$ is called \defemph{open-closed}.
    The representative for $(\del{\gamma},\del{\eta})$
    is a tuple of chains $(A,z)$  with $z=\partial(A)$,
    where $A$ is a $(\Dim+1)$-chain
    created by $\del{\gamma}$
    and $z$ is a $\Dim$-chain 
    created by $\del{\eta}$.

    \item A $\Dim$-pair of the form $(\add{\gamma},\del{\eta})$ is called \defemph{closed-closed}.
    The representative for $(\add{\gamma},\del{\eta})$
    is a tuple of chains $(z,A,z')$ with $z+z'=\partial(A)$,
    where $z$ is a $\Dim$-chain 
    created by $\add{\gamma}$
    and $z'$ is a $\Dim$-chain
    created by $\del{\eta}$.
    Notice that $A$ can contain any $(\Dim+1)$-cells in $\Ucal$.
\end{itemize}
\end{definition}
\begin{remark}
Representatives defined above
are a special case of general zigzag representatives
defined in~\cite{maria2014zigzag}.
For example, in 
a representative $(z,A,z')$ for a closed-closed pair,
having $z+z'=\partial(A)$  ensures 
that 
$[z]=[z']$
in $\Hm_*(\ucplx_{\simpcnt})$,
and the definition in~\cite{maria2014zigzag}
has similar requirements.
Notice that
we may have $z=z'$,
for which  $A$ is simply empty.
\end{remark}

\begin{remark}\label{rmk:constant-bars}
To perform the update, we shall maintain
a representative for each pair in the up-down filtration 
that we work on,
which consists of at most three chains.
This is only a slight difference from the two chains per bar
maintained in~\cite{cohen2006vines},
where one chain is from the `$R$-matrix' and
another is from the `$V$-matrix'.
\end{remark}

We present below a fact helpful for the proof of correctness of algorithms in this paper:
\begin{proposition}
\label{prop:pair-correct}
Let $\pi$ be a set of pairs
of
additions and deletions
in $\Ucal$ s.t.\ each addition and deletion in $\Ucal$
appears exactly once in $\pi$.
If 
each pair in $\pi$ admits a representative,
then the pairs in $\pi$ generate
the intervals in
$\Pers_*(\Ucal)$.

\end{proposition}
\begin{proof}
First consider the first half $\Ucal_u$ of $\Ucal$
containing only additions,
which is indeed a non-zigzag filtration.
We have that the representatives for the pairs in $\pi$
induce an `$RU$' decomposition for the persistence boundary matrix of $\Ucal_u$.
We also have similar facts for the second half $\Ucal_d$ of $\Ucal$
containing only deletions.
Then, for each pair $(\addel{\varsG}, \addel{\varsG'})$ in $\pi$,
we have that $\addel{\varsG}$ is positive and
$\addel{\varsG'}$ is negative
based on the Pairing Uniqueness Lemma in~\cite{cohen2006vines}.
The proposition then follows from 
Proposition~9 of~\cite{dey2021graph}.
\end{proof}

\subsection{Using conversion for easier updates}
\label{sec:ez-up}

Utilizing the conversion  in Section~\ref{sec:fzz-up},
our algorithms work on
the corresponding up-down filtrations of $\Fcal$ and $\Fcal'$
for  the operations  in  Section~\ref{sec:update-oper}.
Specifically, we maintain 
 the barcode
and  representatives
for the corresponding up-down filtration,
and 
the  barcode
for the original  zigzag filtration
can be derived from the bijection $\phi$ as in Theorem~\ref{thm:fzz-intv-map}.
Notice that representatives 
for the pairs in the up-down filtration
are essential for  performing the update.
Moreover, 
by Proposition~\ref{prop:pair-correct},
the correctness of our update algorithms follows from
the correctness of the representatives that we maintain
during the update operations.
We briefly describe the ideas for the update in this section
and provide full details in Appendix~\ref{sec:ez-up-details}.

\paragraph{Outward/inward switch.}
Since the switch in Equation~(\ref{eqn:out-in-sw})
swaps two inclusions of different directions,
the corresponding up-down filtrations before and after
the switch are the same (see Section~\ref{sec:fzz-up}).
So the time complexity for the update is $O(1)$.

\paragraph{Forward/backward switch.}
Corresponding to a forward (resp.\ backward) switch on the original zigzag filtration,
there is a forward (resp.\ backward) switch 
in the up-down filtration~\cite{dey2023revisiting}.
To perform the update for forward/backward switches
on up-down filtrations,
we utilize the algorithm for the {transposition diamond}
in~\cite{maria2014zigzag},
which takes $O(\filtcnt)$ time.
Observe that the algorithm 
in~\cite{cohen2006vines}
cannot be used on up-down filtrations.

\paragraph{Inward expansion.}
For the inward expansion from $\Fcal'$ to $\Fcal$ in Equation~(\ref{eqn:in-contrac-expan}),
assume that the up-down filtration $\Ucal'$ corresponding to $\Fcal'$
has the following 
$\filtcnt-2$ additions and deletions:
\begin{equation*}
\Ucal':
\bullet
\inctosp{\usimp_0}
\bullet
\inctosp{\usimp_{1}}
\bullet
\cdots
\bullet\inctosp{\usimp_{\simpcnt-2}} 
\bullet
\bakinctosp{\usimp_{\simpcnt-1}}
\bullet
\bakinctosp{\usimp_{\simpcnt}}
\bullet
\cdots
\bullet
\bakinctosp{\usimp_{\filtcnt-3}}
\bullet
\end{equation*}
Let $\hat{\sG}$ be the $\DG$-cell corresponding to 
the addition of the simplex $\sG$ to $K_{i-1}$ in $\Fcal$.
We then insert the addition and deletion of $\hat{\sG}$
in the middle of $\Ucal'$ to get
an up-down filtration
$\Ucal^+$:
\begin{equation}
\label{eqn:U-plus}
\Ucal^+:
\bullet
\inctosp{\usimp_0}
\bullet
\inctosp{\usimp_{1}}
\bullet
\cdots
\bullet\inctosp{\usimp_{\simpcnt-2}} 
\bullet
\inctosp{\hat{\sG}} 
\bullet
\bakinctosp{\hat{\sG}}
\bullet
\bakinctosp{\usimp_{\simpcnt-1}}
\bullet
\bakinctosp{\usimp_{\simpcnt}}
\bullet
\cdots
\bullet
\bakinctosp{\usimp_{\filtcnt-3}}
\bullet
\end{equation}
To get 
the barcode and representatives for $\Ucal^+$
from those for $\Ucal'$, 
we utilize the algorithm
for the {injective} and {surjective diamond}
in~\cite{maria2014zigzag},
which takes $O(\filtcnt^2)$ time.
After this, we perform forward and backward switches on $\Ucal^+$ to 
place the two arrows $\inctosp{\hat{\sG}}$ and $\bakinctosp{\hat{\sG}}$
in  appropriate positions to form $\Ucal$,
which is the up-down filtration corresponding to $\Fcal$.
Since no more than $O(\filtcnt)$ switches are performed,
the switches in total
takes $O(\filtcnt^2)$ time.
Therefore, the update for inward expansion takes $O(\filtcnt^2)$ time.
Notice that the above transition from $\Ucal'$ to $\Ucal^+$
and from $\Ucal^+$ to $\Ucal$ is valid because $\hat{\sG}$ has no cofaces
in $\Ucal$ based on the conversion in~\cite{dey2022fast}.

\paragraph{Computing representatives for an initial up-down filtration.}

In order to perform a sequence of updates starting from  an initial zigzag filtration,
we need to compute the barcode and representatives 
for the initial up-down filtration.
This can be done by adapting the algorithm in~\cite{maria2014zigzag}
which takes $O(\filtcnt^3)$ time
for an initial filtration of length $\filtcnt$.

\section{Outward contraction}
\label{sec:fzz-out-contra}

Recall that an outward contraction is the following operation:
\begin{equation}
\label{eqn:out-contrac-in-sec}
\begin{tikzpicture}[baseline=(current  bounding  box.center)]
\node (a) at (0,0) {$\Fcal: K_0 \leftrightarrow
\cdots\leftrightarrow K_{i-2}
\leftrightarrow 
K_{i-1}\bakinctosp{\sG} 
K_i 
\inctosp{\sG} K_{i+1}
\leftrightarrow K_{i+2}\leftrightarrow
\cdots \leftrightarrow K_\filtcnt$}; 
\node (b) at 
(0,-1.02){$\Fcal': K_0 \leftrightarrow
\cdots\leftrightarrow K_{i-2}
\leftrightarrow 
K'_{i}
\leftrightarrow K_{i+2}\leftrightarrow
\cdots \leftrightarrow K_\filtcnt$};
\path[->] ([xshift=7pt]a.south) edge[double distance=2pt,arrows=-{Classical TikZ Rightarrow[length=4pt]}]  ([xshift=7pt]b.north);
\end{tikzpicture}
\end{equation}
where $K'_{i}=K_{i-1}=K_{i+1}$.
We also assume that  $\sG$ is a $\Dim$-simplex.
Notice that the indices for $\Fcal'$ are not consecutive,
i.e., $i-1$ and $i+1$ are skipped.

\paragraph{Adjacency change.} 
We first describe the challenge in updating the barcodes.
Let $\Ucal$ and $\Ucal'$ be the up-down filtrations
 corresponding to $\Fcal$ and $\Fcal'$ respectively.
According to~\cite{dey2022fast},
the first step of the conversion from 
$\Fcal$
to 
$\Ucal$ is to treat
each repeatedly added simplex in $\Fcal$
as a new copy of a $\Delta$-cell.
This results in  a zigzag filtration of $\Delta$-complexes which we name as $\hat{\Fcal}$.
Notice that $\Fcal$ and $\hat{\Fcal}$ 
are still isomorphic
(see the example in Figure~\ref{fig:fzz-out-contrac}).
However, by treating simplices as distinct  $\Delta$-cells in $\hat{\Fcal}$,
one can achieve the conversion to the up-down filtration~\cite{dey2022fast}.
Since every simplex added in $\Fcal$ must be deleted later,
let $K_{j}\inctosp{\sG} K_{j+1}$
be the addition associated to the deletion $K_{i-1}\bakinctosp{\sG}K_{i}$
and let $K_k \bakinctosp{\sG}K_{k+1}$ 
be the deletion associated    to the  addition $K_i \inctosp{\sG} K_{i+1}$
in $\Fcal$:
\begin{equation*}%
\begin{tikzpicture}[baseline=(current  bounding  box.center)]
\node (a) at (0,0) {$\Fcal: 
\cdots\leftrightarrow
K_j \inctosp{\sG}K_{j+1} \leftrightarrow
\cdots
\leftrightarrow 
K_{i-1}\bakinctosp{\sG} 
K_i 
\inctosp{\sG} K_{i+1}
\leftrightarrow 
\cdots
\leftrightarrow
K_k \bakinctosp{\sG}K_{k+1}
\leftrightarrow
\cdots 
$}; 
\node (b) at (0,-1.1){$\hat{\Fcal}: 
\cdots\leftrightarrow
\hat{K}_j \inctosp{\sG_1}\hat{K}_{j+1} \leftrightarrow
\cdots
\leftrightarrow 
\hat{K}_{i-1}\bakinctosp{\sG_1} 
\hat{K}_i 
\inctosp{\sG_2} \hat{K}_{i+1}
\leftrightarrow 
\cdots
\leftrightarrow
\hat{K}_k \bakinctosp{\sG_2}\hat{K}_{k+1}
\leftrightarrow
\cdots 
$}; 
\path[->] ([xshift=7pt]a.south) edge[double distance=2pt,arrows=-{Classical TikZ Rightarrow[length=4pt]}]  ([xshift=7pt]b.north);
\end{tikzpicture}
\end{equation*}
In the above sequences,
when $\sG$ is added to $K_j$,
we denote the corresponding $\Delta$-cell in $\hat{\Fcal}$
as $\sG_1$, and 
when $\sG$ is added to $K_{i}$,
we denote the corresponding $\Delta$-cell in $\hat{\Fcal}$
as $\sG_2$.
Let $\SG_1$ be the set of $(\Dim+1)$-cofaces
of $\sG$ added between 
$K_{j}\inctosp{\sG} K_{j+1}$
and $K_{i-1}\bakinctosp{\sG}K_{i}$ in $\Fcal$,
and let $\SG_2$ be the set of $(\Dim+1)$-cofaces
of $\sG$ added between 
$K_{i}\inctosp{\sG} K_{i+1}$
and $K_{k}\bakinctosp{\sG}K_{k+1}$ in $\Fcal$.
Furthermore, 
let $\hat{\SG}_1,\hat{\SG}_2$ be the sets
of $\Delta$-cells in $\hat{\Fcal}$ corresponding to
${\SG}_1,{\SG}_2$ respectively.
Since adjacency of simplices/cells in $\Fcal$ and $\hat{\Fcal}$
are the same,
we have that 
cells in $\hat{\SG}_1$
have $\sG_1$ in their boundaries
and 
cells in $\hat{\SG}_2$
have $\sG_2$ in their boundaries.

Now consider $\Fcal'$ and the corresponding $\hat{\Fcal}'$
after the contraction:
\begin{equation*}
\begin{tikzpicture}[baseline=(current  bounding  box.center)]
\node (a) at (0,0) {$\Fcal': 
\cdots\leftrightarrow
K_j \inctosp{\sG}K_{j+1} \leftrightarrow
\cdots
\leftrightarrow 
K'_i 
\leftrightarrow 
\cdots
\leftrightarrow
K_k \bakinctosp{\sG}K_{k+1}
\leftrightarrow
\cdots 
$}; 
\node (b) at (0,-1.1){$\hat{\Fcal}': 
\cdots\leftrightarrow
\hat{K}_j \inctosp{\sG_0}\hat{K}_{j+1} \leftrightarrow
\cdots
\leftrightarrow 
\hat{K}'_i 
\leftrightarrow 
\cdots
\leftrightarrow
\hat{K}_k \bakinctosp{\sG_0}\hat{K}_{k+1}
\leftrightarrow
\cdots 
$}; 
\path[->] ([xshift=8pt]a.south) edge[double distance=2pt,arrows=-{Classical TikZ Rightarrow[length=4pt]}]  ([xshift=8pt]b.north);
\end{tikzpicture}
\end{equation*}
Because of the contraction,
$K_{j}\inctosp{\sG} K_{j+1}$
and $K_{k}\bakinctosp{\sG}K_{k+1}$ 
become  associated addition and deletion of $\sG$ in $\Fcal'$.
For clarity, we name the $\Delta$-cell
in $\hat{\Fcal}'$ corresponding to the $\sG$ added in $K_j$
as $\sG_0$.
Based on the conversion in~\cite{dey2022fast},
we have that after the contraction, cells in
$\hat{\SG}_1\union\hat{\SG}_2$ now both have $\sG_0$ in their boundaries
in $\hat{\Fcal}'$ (see $\tau_1,\tau_2$ in Figure~\ref{fig:fzz-out-contrac} for an example).
We also notice that $\sG_1$ and $\sG_2$ have the same boundary
because after  $\sG$ is added to $K_j$ in $\Fcal$,
the boundary simplices of $\sG$ are never deleted before
$\sG$ is deleted from $K_k$
(in general, different $\DG$-cells in the converted up-down filtration corresponding to the same simplex
could have different boundaries).
The change can then be formally stated as follows:
\begin{observation}
After the outward contraction,
the boundary $\Dim$-cell $\sG_1$ of $(\Dim+1)$-cells in $\hat{\SG}_1$ and
the boundary $\Dim$-cell $\sG_2$ of $(\Dim+1)$-cells in $\hat{\SG}_2$
are identified as the same $\Dim$-cell $\sG_0$ in $\hat{\Fcal}'$ (and $\Ucal'$).
\end{observation}
While  the cell identification seems to have no effect in $\hat{\Fcal}$ and $\hat{\Fcal}'$ 
as in Figure~\ref{fig:fzz-out-contrac}
because $\sG_1$ and $\sG_2$ never appear in 
the same $\Delta$-complex,
the effect is evident  in the up-down
filtrations $\Ucal$ and $\Ucal'$
(which we work on)
where all cells eventually appear in the same $\DG$-complex.
In Figure~\ref{fig:fzz-out-contrac},
$\sG_1$ is a boundary edge of
$\tau_1$ and $\sG_2$ is a boundary edge of
$\tau_2$ in $\Ucal$.
After the contraction,
both $\tau_1$ and $\tau_2$ have $\sG_0$ in their boundaries in $\Ucal'$.
To better understand the conversion, 
we notice that  
to build $\Ucal$ 
from $\hat{\Fcal}$, 
one only needs to 
list all the additions first and then the deletions in $\hat{\Fcal}$,
following their orders in $\hat{\Fcal}$~\cite{dey2022fast}
(conversion from  $\hat{\Fcal}'$ to $\Ucal'$ is the same);
see also Section~\ref{sec:fzz-up}.

\begin{figure}[!tb]
  \centering
  \includegraphics[width=\linewidth]{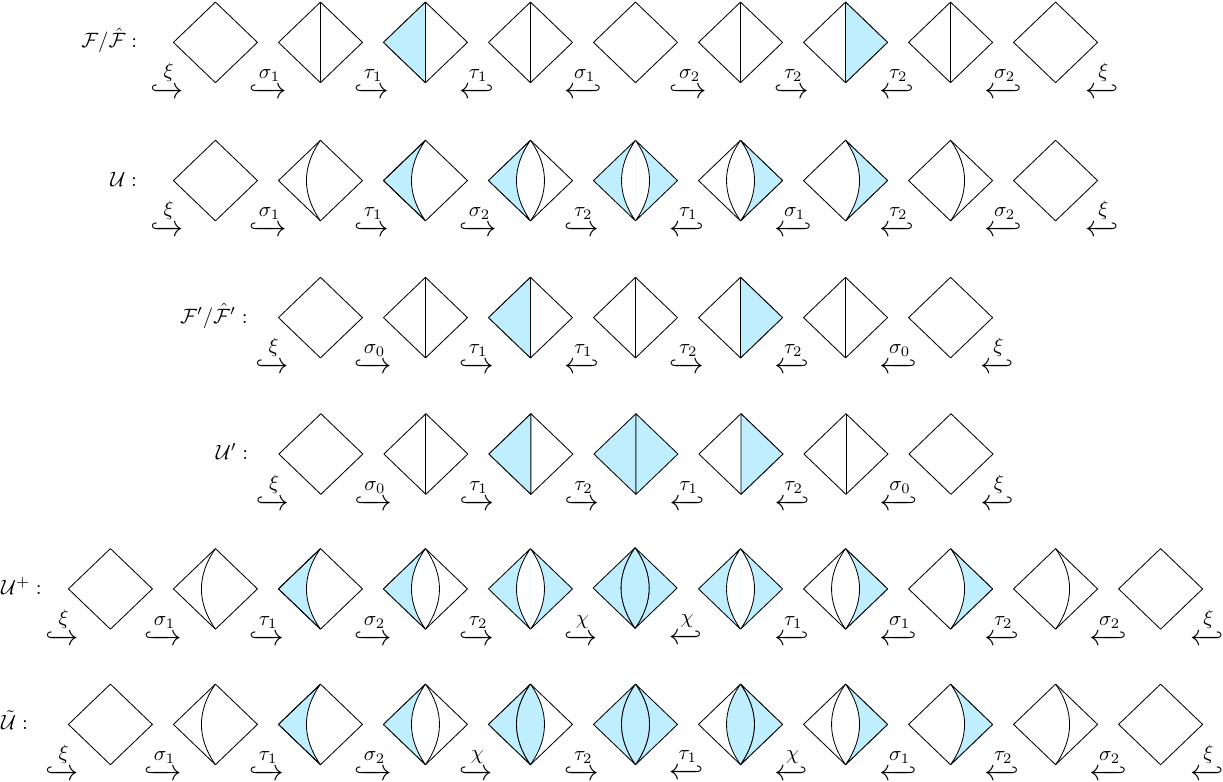}
  \caption{An example illustrating the adjacency change of $\Delta$-cells
  and the algorithm ideas
  for outward contraction.
  Complexes in
  $\Fcal$ and $\hat{\Fcal}$ are 
  isomorphic
  and so
  $\Fcal$ and $\hat{\Fcal}$
  ($\Fcal'$ and $\hat{\Fcal}'$ as well)
  are combined together where the arrows are labelled by the 
  $\Delta$-cells.
  For simplicity, 
  the filtrations
  do not start and end with empty complexes.}
  \label{fig:fzz-out-contrac}
\end{figure}

\begin{figure}[!tb]
  \includegraphics[width=0.3\linewidth]{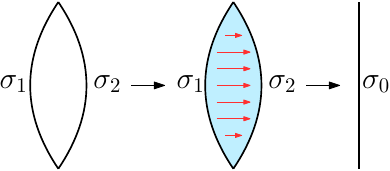}
  \caption{Attaching a $(\Dim+1)$-cell $\extsimp$
with boundary being $\sG_1+\sG_2$
and `collapsing' $\extsimp$ to identify $\sG_1$ and $\sG_2$.}
  \label{fig:collapse}
\end{figure}

\paragraph{Key ideas for the computation.}
We need to identify $\sG_1,\sG_2$ in $\Ucal$ to obtain $\Ucal'$. This
can be alternatively done by attaching an additional $(\Dim+1)$-cell $\extsimp$
whose boundary equals $\sG_1+\sG_2$
and then `collapsing' $\extsimp$ onto one of its boundary $\Dim$-cell
(see Figure~\ref{fig:collapse} for $\Dim=1$).
Therefore, our solution is that we first perform an inward expansion 
in the middle of $\Ucal$
with 
$\add{\extsimp}$ and $\del{\extsimp}$
added,
to get a new  up-down filtration  $\Ucal^+$ (as shown in Figure~\ref{fig:fzz-out-contrac}).
Then, we apply several forward and backward switches
to make 
$\add{\extsimp}$ 
immediately after 
$\add{\sG_2}$
and make 
$\del{\extsimp}$ 
immediately before 
$\del{\sG_1}$,
getting an up-down filtration  $\tilde{\Ucal}$ (see Figure~\ref{fig:fzz-out-contrac}).
The reason for constructing the additional filtration $\Ucal^+$ and then $\tilde{\Ucal}$ is 
that 
it is not obvious how to compute $\Pers_*({\Ucal}')$ directly in an efficient way
but $\Pers_*(\tilde{\Ucal})$ can be computed from $\Pers_*({\Ucal})$ in $O(\filtcnt^2)$ time
(see the algorithm for inward expansion).
We observe that $\Pers_*(\tilde{\Ucal})$ and $\Pers_*({\Ucal}')$ are `almost' the same 
(see Proposition~\ref{prop:pers-Utild-Uprime-almost-same} for a formal statement),
which leads to an $O(\filtcnt^2)$ algorithm for computing $\Pers_*({\Ucal}')$.
To see  the relevance of $\Pers_*(\tilde{\Ucal})$ and $\Pers_*({\Ucal}')$, 
first
consider the two \emph{non-zigzag} filtrations forming $\Ucal$:
\begin{align*}
&\Ucal_u:
L^u_0\incto
\cdots\incto
L^u_s\inctosp{\sG_1}
L^u_{s+1}
\incto
\cdots
\inctosp{\mu}
L^u_t\inctosp{\sG_2}
L^u_{t+1}
\inctosp{\mu'}
L^u_{t+2}\incto
\cdots
\incto
L^u_\simpcnt,\\
&\Ucal_d:
L^d_0\incto
\cdots\incto
L^d_x\inctosp{\sG_2}
L^d_{x+1}
\incto
\cdots
\inctosp{\rho}
L^d_y\inctosp{\sG_1}
L^d_{y+1}
\inctosp{\rho'}
L^d_{y+2}\incto
\cdots
\incto
L^d_\simpcnt.
\end{align*}
In the above sequences, 
$\Ucal_u$ is the \emph{ascending} part of $\Ucal$
consisting of all the additions
and $\Ucal_d$ is the \emph{descending} part of $\Ucal$
consisting of all the deletions
so that $\Ucal$ equals $\Ucal_u$ concatenated with the reversed $\Ucal_d$.
Similarly, $\Ucal'$ can be represented as follows:
\begin{align*}
&\Ucal'_u:
L^u_0\incto
\cdots\incto
L^u_s\inctosp{\sG_0}
L'^u_{s+1}
\incto
\cdots\inctosp{\mu}
L'^u_t
\inctosp{\mu'} 
L'^{u}_{t+2}
\incto\cdots
\incto
L'^{u}_\simpcnt,\\
&\Ucal'_d:
L^d_0\incto
\cdots\incto
L^d_x\inctosp{\sG_0}
L'^d_{x+1}
\incto
\cdots
\inctosp{\rho}
L'^{d}_y\inctosp{\rho'}
L'^{d}_{y+2}
\incto
\cdots
\incto
L'^{d}_\simpcnt.
\end{align*}
By the conversion from $\hat{\Fcal}$ to $\Ucal$ and the conversion from $\hat{\Fcal}'$ to $\Ucal'$~\cite{dey2022fast}
described in Section~\ref{sec:fzz-up},
the only difference of the sequences of additions in  $\Ucal_u$ and $\Ucal'_u$
is that 
the addition of ${\sG_2}$ 
disappears in  $\Ucal'_u$
and $\sG_1$ in $\Ucal_u$ is renamed as $\sG_0$ in $\Ucal'_u$.
Therefore, the only difference of $L'^u_a$ and $L^u_a$ for $s+1\leq a\leq t$
is the renaming of $\sG_1$ to $\sG_0$ in $L'^u_a$.
Furthermore,
because of the adjacency change,
$L'^u_a$ for $a\geq t+2$ can be considered as derived from
$L^u_a$ by identifying $\sG_2$ to $\sG_1$ (renamed as $\sG_0$).
Notice that we have symmetric changes from $\Ucal_d$ to $\Ucal'_d$,
i.e., the addition of
${\sG_1}$ 
disappears and
$\sG_1$ is identified to $\sG_2$ which is renamed as $\sG_0$.

Moreover, we write $\tilde{\Ucal}$  as follows:
\begin{align*}
&\tilde{\Ucal}_u:L^u_0\incto
\cdots
\incto
L^u_s\inctosp{\sG_1}
L^u_{s+1}
\cdots
\inctosp{\mu}
L^u_t\inctosp{\sG_2}
\tilde{L}^u_{t+1}\inctosp{\extsimp}
\tilde{L}^u_{t+2}
\inctosp{\mu'}
\tilde{L}^u_{t+3}
\incto\cdots
\incto
\tilde{L}^u_{\simpcnt+1},\\
&\tilde{\Ucal}_d:L^d_0\incto
\cdots\incto
L^d_x\inctosp{\sG_2}
L^d_{x+1}
\cdots
\inctosp{\rho}
L^d_y\inctosp{\sG_1}
\tilde{L}^d_{y+1}\inctosp{\extsimp}
\tilde{L}^d_{y+2}
\inctosp{\rho'}
\tilde{L}^d_{y+3}
\incto
\cdots\incto
\tilde{L}^d_{\simpcnt+1}.
\end{align*}

We then observe the following:
In $\tilde{\Ucal}_u$, the effect of 
the addition of $\sG_2$
is immediately annulled due to the subsequent addition of $\extsimp$
because $\sG_2$ can then be identified with $\sG_1$ by collapsing
$\extsimp$.
Specifically, we have that $\tilde{L}^u_a$
is homotopy equivalent to $L'^u_{a-1}$
for each   $a\geq t+3$ (see Figure~\ref{fig:fzz-out-contrac} for an example).
Similarly, in $\tilde{\Ucal}_d$, the effect of the addition of $\sG_1$
is  annulled due to the addition of $\extsimp$
and we have that $\tilde{L}^d_a$
is homotopy equivalent to $L'^d_{a-1}$
for each   $a\geq y+3$.
Furthermore, by the nature of persistence~\cite{edelsbrunner2000topological},
$\sG_2$ must be paired with $\extsimp$ in $\tilde{\Ucal}_u$
(because adding  $\sG_2$  creates a cycle $\sG_1+\sG_2$
which is immediately killed by adding $\extsimp$)
and $\sG_1$ must be paired with $\extsimp$ in $\tilde{\Ucal}_d$ (with a similar reason).
We now have:
\begin{definition}
For each  addition or deletion $\addel{\eta}$ in $\tilde{\Ucal}$
s.t.\ $\addel{\eta}
\not\in
\{
\add{\sG_2},\,\add{\extsimp},\,\del{\extsimp},
\del{\sG_1}
\}
$,
define its corresponding 
addition/deletion
in ${\Ucal}'$,
denoted $\theta(\addel{\eta})$,
as follows:
$\theta(\add{\sG_1})=\,\add{\sG_0}$;
$\theta(\del{\sG_2})=\,\del{\sG_0}$;
for any remaining $\addel{\eta}$,
the corresponding $\theta(\addel{\eta})$ is the same.
\end{definition}
\begin{proposition}
\label{prop:pers-Utild-Uprime-almost-same}
Given $\Pers_*(\tilde{\Ucal})$, 
one only needs to do the following to obtain  $\Pers_*({\Ucal}')$:
    Ignoring the pairs $(\add{\sG_2},\add{\extsimp})$
    and $(\del{\extsimp},\del{\sG_1})$ in 
    $\Pers_*(\tilde{\Ucal})$,
    for each remaining pair $(\addel{\eta},\addel{\gamma})\in\Pers_*(\tilde{\Ucal})$, 
    produce a corresponding pair
    $(\theta(\addel{\eta}),\thG(\addel{\gamma}))\in\Pers_*({\Ucal'})$.
\end{proposition}
\begin{proof}
As mentioned,
we have the following equivalence of homotopy types 
between complexes in $\tilde{\Ucal}_u$  and $\Ucal_u'$:
    (i) complexes at indices from $0$ to $s$ in $\tilde{\Ucal}_u$  and $\Ucal_u'$
    are the same;
    (ii) $L^u_a$  and $L'^u_a$ are the same complex for $s+1\leq a\leq t$
with the only difference on the naming of $\sG_0$ and $\sG_1$;
    (iii) $\tilde{L}^u_a$
is homotopy equivalent to $L'^u_{a-1}$
for $a\geq t+3$.
Therefore, if we view  
$L^u_t\incto\tilde{L}^u_{t+3}$ as a single inclusion 
in $\tilde{\Ucal}_u$,
the induced zigzag modules $\Hm_*(\tilde{\Ucal}_u)$ 
and $\Hm_*({\Ucal}'_u)$ are isomorphic.
We also have similar results for 
$\Hm_*(\tilde{\Ucal}_d)$ 
and $\Hm_*({\Ucal}'_d)$.
Hence, 
$\Hm_*(\tilde{\Ucal})$ 
and $\Hm_*({\Ucal}')$ are isomorphic
(by treating $L^u_t\incto\tilde{L}^u_{t+3}$
and $L^d_y\incto\tilde{L}^d_{y+3}$ as single inclusions).
The proposition then follows.
\end{proof}
\begin{remark*} 
An alternative proof of Proposition~\ref{prop:pers-Utild-Uprime-almost-same}
follows from Proposition~\ref{prop:pair-correct} and the correctness of the representatives
that we obtain for $\Pers_*({\Ucal}')$ (see Section~\ref{sec:out-contrac-rewrite-rep}
for details).
\end{remark*}
\begin{example*}
Ignoring the two pairs without correspondence,
we have the following mapping of pairs from $\Pers(\tilde{\Ucal})$
to $\Pers({\Ucal}')$ for the example in Figure~\ref{fig:fzz-out-contrac} ($\varsG$ 
is a fixed edge in the boundary of the diamond):
$(\add{\varsG},\add{\tau_2})\mapsto(\add{\varsG},\add{\tau_2}),
(\add{\sG_1},\add{\tau_1})\mapsto(\add{\sG_0},\add{\tau_1}),
(\del{\tau_1},\del{\varsG})\mapsto(\del{\tau_1},\del{\varsG}),
(\del{\tau_2},\del{\sG_2})\mapsto(\del{\tau_2},\del{\sG_0})$.
\end{example*}

\paragraph{Algorithm.}
From  the maintained $\Pers_*(\Ucal)$
and its  representatives,
we first compute $\Pers_*(\tilde{\Ucal})$ and the representatives
in $O(\filtcnt^2)$ time.
Based on Proposition~\ref{prop:pers-Utild-Uprime-almost-same}, 
$\Pers_*(\Ucal')$ can then be easily derived.
To get the representatives for $\Pers_*(\Ucal')$, 
we notice that chains and cycles in $\tilde{\Ucal}$
can be naturally `re-written' as chains and cycles in $\Ucal'$,
e.g.,
the chain $\tau_1+\tau_2+\extsimp$ in $\tilde{\Ucal}$
can be re-written as $\tau_1+\tau_2$ in $\Ucal'$
where their boundary stays the same.
We  provide in Section~\ref{sec:out-contrac-rewrite-rep} the details of re-writing
the representatives for $\Pers_*(\tilde{\Ucal})$
to get the representatives for $\Pers_*(\Ucal')$.
Since the representative re-writing in Section~\ref{sec:out-contrac-rewrite-rep}
takes $O(\filtcnt^2)$ time
(there are no more than $O(\filtcnt)$
representatives needed to be re-written
and each  takes $O(\filtcnt)$ time), 
we conclude the following:
\begin{theorem}
Zigzag barcode for an outward contraction can
be updated in $O(\filtcnt^2)$ time.
\end{theorem}

\subsection{Details of representative re-writing}
\label{sec:out-contrac-rewrite-rep}
Notice that representatives for $q$-pairs with $q<\Dim-1$ or $q>\Dim+1$ do not change
from $\tilde{\Ucal}$ to $\Ucal'$
because both the cells  and their adjacency in the representatives stay 
the same.

\paragraph{$(\Dim-1)$-pairs.}
For a $(\Dim-1)$-pair  in $\Pers_*(\tilde{\Ucal})$,
we only need to re-write its representative 
when the $\Dim$-chain (denoted as $A$) in the representative
contains $\sG_1$ or $\sG_2$.
If $A$ contains both $\sG_1$ and $\sG_2$,
then we can delete $\sG_1$ and $\sG_2$ from $A$
and the boundary of $A$ does not change ($\partial(\sG_1+\sG_2)=0$).
Notice that 
if the $(\Dim-1)$-pair is closed-open,
the creator of $A$ does not change 
by deleting $\sG_1$ and $\sG_2$ 
because $A$ is not created by $\sG_1$ ($\sG_2$ is added later than $\sG_1$)
and is not created by $\sG_2$ ($\add{\sG_2}$ forms a $\Dim$-pair 
as in Proposition~\ref{prop:pers-Utild-Uprime-almost-same}).
Similarly,
if the $(\Dim-1)$-pair is open-closed,
the creator of $A$ also does not change.
If $A$ contains only $\sG_1$ or only $\sG_2$,
then we can replace  $\sG_1$ or $\sG_2$ in $A$ by $\sG_0$
without changing the boundary and the  creator of $A$.
In all cases,
the resulting representative is a valid one for the same pair in $\Pers_*(\Ucal')$.

\paragraph{$\Dim$-pairs.}
First consider a  closed-open $\Dim$-pair $(\add{\eta},\add{\gamma})\in\Pers_*(\tilde{\Ucal})$
with a representative  $(z,A)$,
where $\eta$ is a $\Dim$-cell and $\gamma$ is a $(\Dim+1)$-cell.
Since we only consider pairs in $\Pers_*(\tilde{\Ucal})$ having a correspondence in $\Pers_*(\Ucal')$,
we assume that $(\add{\eta},\add{\gamma})\neq(\add{\sG_2},\add{\extsimp})$.
Denote the corresponding pair in 
$\Pers_*(\Ucal')$ as $(\add{\eta'},\add{\gamma})$,
where $\eta'\neq\eta$ only when ${\eta}={\sG_1}$ (see Proposition~\ref{prop:pers-Utild-Uprime-almost-same}).
We have the following cases:
\begin{enumerate}
    \item\label{itm:sG1-nin-z-sG2-nin-z-oG-nin-A} 
    $\extsimp\not\in A$ and $\sG_1,\sG_2\not\in z$:
    Since $\sG_1\not\in z$, the creator of $z$ cannot be $\sG_1$
    and therefore $\eta\neq\sG_1$.
    So we have $\eta'=\eta$.
    Notice that $A$ may contain some cofaces of $\sG_1$ or $\sG_2$
    but the number of cofaces of $\sG_1$ (resp.\ $\sG_2$) in $A$ must be even.
    So if we consider $A$ as a  chain in $\Ucal'$,
    $A$ must contain an even number of cofaces of $\sG_0$.
    This means that the boundary of $A$  also equals $z$ in $\Ucal'$
    (i.e., we consider $A$ and $z$ as chains in $\Ucal'$).
    By Definition~\ref{dfn:ud-rep}, $(z,A)$ is also a  representative for 
    $(\add{\eta'},\add{\gamma})\in\Pers_*(\Ucal')$.
    
    \item\label{itm:sG1-in-z-sG2-nin-z-oG-nin-A} 
    $\extsimp\not\in A$, $\sG_1\in z$, and $\sG_2\not\in z$:
    $A$ contains an odd number of cofaces of $\sG_1$
    and an even number of cofaces of $\sG_2$.
    So the chain $A$ in $\Ucal'$ contains an odd number of cofaces of $\sG_0$.
    Let the boundary of $A$ in $\Ucal'$ be $z'$.
    Then, the only difference of $z$ and $z'$ is that $\sG_1$ in $z$
    is replaced by $\sG_0$ in $z'$.
    Since  the addition of $\sG_1$ in $\tilde{\Ucal}$ is at the same position as
     the addition of $\sG_0$ in $\Ucal'$,
    we have that $z'$ is also  created by
    $\eta'$
    in $\Ucal'$.
    So $(z',A)$ is a representative for 
    $(\add{\eta'},\add{\gamma})\in\Pers_*(\Ucal')$.
    
    \item\label{itm:sG1-nin-z-sG2-in-z-oG-nin-A} 
    $\extsimp\not\in A$, $\sG_1\not\in z$, and $\sG_2\in z$:
    Since $\sG_1\not\in z$, we have $\eta'=\eta$ (see Case~\ref{itm:sG1-nin-z-sG2-nin-z-oG-nin-A}).
    Let the boundary of $A$ in $\Ucal'$ be $z'$.
    Then, the only difference of $z$ and $z'$ is that $\sG_2$ in $z$
    is replaced by $\sG_0$ in $z'$ (see Case~\ref{itm:sG1-in-z-sG2-nin-z-oG-nin-A}).
    Since  $\sG_2\in z$ and the creator of $z$   is $\eta\neq\sG_2$,
    the addition of $\sG_2$
    must be before the addition of $\eta$
    in $\tilde{\Ucal}$.
    This implies that
    the addition of $\sG_0$ 
    is also before the addition of $\eta'=\eta$
    in $\Ucal'$,
    which indicates that $z'$ is also created by $\eta'$.
    Hence, $(z',A)$ is a representative for 
    $(\add{\eta'},\add{\gamma})\in\Pers_*(\Ucal')$.
    
    \item $\extsimp\not\in A$ and $\sG_1,\sG_2\in z$:
    We have 
    $\eta\neq\sG_1$
    because $\sG_2\in z$ is added after $\sG_1$.
    We hence  have $\eta'=\eta$.
    Let the boundary of $A$ in $\Ucal'$ be $z'$.
    Since $\sG_1$ and $\sG_2$
    are identified as the same cell in $\Ucal'$, 
    $\sG_1$ and $\sG_2$  are cancelled out in $z'$.
    Since $\eta\neq\sG_1$ and $\eta\neq\sG_2$,
    we have that $z'$ is also created by $\eta'=\eta$.
    So $(z',A)$ is a representative for 
    $(\add{\eta'},\add{\gamma})\in\Pers_*(\Ucal')$.
    
    \item $\extsimp\in A$:
    Let $A'=A\setminus\Set{\extsimp}$
    and $z'=z+\sG_1+\sG_2$.
    We have that $\partial(A')=\partial(A)+\partial(\extsimp)=z'$.
    Since $\sG_1$ and $\sG_2$ are identified as the same cell
     in $\Ucal'$,
     $\sG_1+\sG_2$  are cancelled out in $z'$ if we consider $z'$
     as a chain in $\Ucal'$.
    Notice that deleting $\extsimp$ from $A$ does not change the creator
    because $\gamma\neq\extsimp$.
    So  $A'$ is also created by $\gamma$.
     The remaining re-writing can then be seen in Case 1--4.

\end{enumerate}

Notice that re-writing representatives for open-closed $\Dim$-pairs in 
$\Pers_*(\tilde{\Ucal})$ is symmetric to the re-writing for closed-open $\Dim$-pairs
described above.
Now consider a  closed-closed $\Dim$-pair $(\add{\eta},\del{\gamma})\in\Pers_*(\tilde{\Ucal})$
with a representative tuple $(z_1,A,z_2)$.
If $\extsimp\not\in A$, re-writing $z_1$ and $z_2$ is similar as for closed-open $\Dim$-pairs
with a possible replacement of $\sG_1$ or $\sG_2$ with $\sG_0$
or a cancel out of $\sG_1+\sG_2$.
If $\extsimp\in A$, then 
let $A'=A\setminus\Set{\extsimp}$.
We have that $\partial(A')=\partial(A)+\partial(\extsimp)=z_1+z_2+\sG_1+\sG_2$.
Since $\sG_1+\sG_2$ are cancelled out  in $\Ucal'$,
the remaining process is the same as for the case of $\extsimp\not\in A$.

\paragraph{$(\Dim+1)$-pairs.}
Let $(z,A)$ be a representative for 
a closed-open 
$(\Dim+1)$-pair $(\add{\eta},\add{\gamma})\in\Pers_*(\tilde{\Ucal})$.
Notice  that the corresponding pair in $\Pers_*({\Ucal}')$
is also $(\add{\eta},\add{\gamma})$ because
$\eta$ is a $(\Dim+1)$-cell
and $\gamma$ is a $(\Dim+2)$-cell (see Proposition~\ref{prop:pers-Utild-Uprime-almost-same}).
We have  $\extsimp\not\in z=\partial(A)$ because $\extsimp$ has no cofaces in $\tilde{\Ucal}$.
So $(z,A)$ is automatically a representative for 
$(\add{\eta},\add{\gamma})\in\Pers_*(\Ucal')$.
Re-writing representatives for 
open-closed 
$(\Dim+1)$-pairs is the same.
Now consider a  closed-closed $(\Dim+1)$-pair $(\add{\eta},\del{\gamma})\in\Pers_*(\tilde{\Ucal})$
with a representative  $(z_1,A,z_2)$.
We have that $\eta\neq\extsimp$, $\gamma\neq\extsimp$,
and the corresponding pair in $\Pers_*({\Ucal}')$
is the same.
The fact that $z_1+z_2=\partial(A)$
implies that $\extsimp\not\in z_1+z_2$.
So $\extsimp$ is in neither of $z_1$, $z_2$
or is in both of them.
If $\extsimp$ is not in the two cycles,
then $(z_1,A,z_2)$ is
automatically a representative for 
the pair $(\add{\eta},\del{\gamma})\in\Pers_*(\Ucal')$.
If $\extsimp$ is in both $z_1$ and $z_2$,
let $z'_1=z_1\setminus\Set{\extsimp}$
and $z'_2=z_2\setminus\Set{\extsimp}$
be two chains in $\tilde{\Ucal}$.
Then, $\partial(z'_1)=\partial(z_1)+\partial(\extsimp)=\sG_1+\sG_2$.
Similarly, $\partial(z'_2)=\sG_1+\sG_2$.
Furthermore, we have $z'_1+z'_2=z_1+z_2=\partial(A)$.
Notice that if we consider $z'_1$ and $z'_2$ as chains in $\Ucal'$,
then $z'_1$ and $z'_2$ become cycles because $\sG_1$ and $\sG_2$
are identified as the same cell in $\Ucal'$.
Also, the creator of $z_1$ and $z'_1$ are the same
and the creator of $z_2$ and $z'_2$ are the same
because $z_1$ and $z_2$ are not created by $\extsimp$.
So $(z'_1,A,z'_2)$
is a representative for 
the pair $(\add{\eta},\del{\gamma})\in\Pers_*(\tilde{\Ucal})$.

\section{Outward expansion}\label{sec:fzz-out-expan} 
An outward expansion is the reverse operation of an outward contraction 
in Equation~(\ref{eqn:out-contrac-in-sec}),
and
 we let $\Fcal$, $\Fcal'$, and $\sG$
be as defined in Equation~(\ref{eqn:out-contrac-in-sec})
throughout the section,
where $\sG$
is also a $\Dim$-simplex.
Notice that 
for outward expansion,
we are given the barcode and representatives for 
$\Ucal'$ (as in Section~\ref{sec:fzz-out-contra}),
and the goal is to compute 
those for 
$\Ucal$.
We observe  the following adjacency change reverse to that for outward contraction:
\begin{observation}
\label{obsv:out-expan-adj-change}
After the outward expansion,
the boundary $\Dim$-cell for some $(\Dim+1)$-cofaces of $\sG_0$
in $\Ucal'$ becomes $\sG_1$ in $\Ucal$,
while the boundary $\Dim$-cell for the other $(\Dim+1)$-cofaces of $\sG_0$
in $\Ucal'$ becomes $\sG_2$ in $\Ucal$.
\end{observation}

We thereby do the following:
\begin{enumerate}
    \item 
    Construct $\tilde{\Ucal}$ (as in Figure~\ref{fig:fzz-out-contrac}) 
    by first
    inserting $\add{\sG_2}$ and $\add{\extsimp}$
    into the ascending part $\Ucal'_u$ of $\Ucal'$;
    the position of the insertion
    can be  derived from the position of the expansion in the original $\Fcal$ and $\Fcal'$.
    Identify $\sG_0$ as the right-most $\DG$-cell in $\Ucal'_u$ 
     corresponding to  $\sG$ 
    added before the newly inserted $\sG_2$.
    Rename $\sG_0$ as $\sG_1$.
    Change
the boundary cell $\sG_0$ (renamed as $\sG_1$)
to $\sG_2$ for those $(\Dim+1)$-cofaces of $\sG_0$
in  $\Ucal'_u$ 
added after the newly inserted  ${\sG_2}$.
This finishes the construction of the ascending part of $\tilde{\Ucal}$.
Constructing the descending part of $\tilde{\Ucal}$
is symmetric.

    \item\label{itm:out-expan-rewrite-step}
    Re-write the representatives for $\Pers_*(\Ucal')$
    to get the representatives for the
    the corresponding pairs in $\Pers_*(\Tilde{\Ucal})$
    (see Proposition~\ref{prop:pers-Utild-Uprime-almost-same}).
We also  let the two extra pairs 
$(\add{\sG_2},\add{\extsimp})$ 
and $(\del{\extsimp},\del{\sG_1})$ have the representative
$(\sG_1+\sG_2,\extsimp)$
and
$(\extsimp,\sG_1+\sG_2)$
respectively.
Details of the representative re-writing
are provided in Section~\ref{sec:out-expan-rep-rewri}.

    \item From $\Pers_*(\tilde{\Ucal})$ and the representatives,
    obtain $\Pers_*({\Ucal})$ and the representatives
    in $O(\filtcnt^2)$ time 
    (see the algorithm for inward contraction
    in Section~\ref{sec:fzz-in-contra}). 
\end{enumerate}

We conclude: 

\begin{theorem}
Zigzag barcode for an outward expansion can
be updated in $O(\filtcnt^2)$ time.
\end{theorem}

\subsection{Details of representative re-writing}
\label{sec:out-expan-rep-rewri}

Notice that similarly as in Section~\ref{sec:out-contrac-rewrite-rep},
representatives for $q$-pairs with $q<\Dim-1$ or $q>\Dim+1$ do not change
from $\Ucal'$ to $\tilde{\Ucal}$.

\paragraph{$(\Dim-1)$-pairs.}
For a $(\Dim-1)$-pair  $(\addel{\eta},\addel{\gamma})\in\Pers_*({\Ucal'})$,
we only need to re-write its representative 
when  the $\Dim$-chain (denoted as $A$) in the representative
contains $\sG_0$.
If $(\addel{\eta},\addel{\gamma})$ is  closed-open or  closed-closed,
we replace $\sG_0$ by $\sG_1$ in $A$;
otherwise, we replace $\sG_0$ by $\sG_2$ in $A$.
It can be verified that the re-written representative is a valid one
for the corresponding pair in $\Pers_*(\tilde{\Ucal})$.

\paragraph{$\Dim$-pairs.}
First consider a  closed-open $\Dim$-pair $(\add{\eta},\add{\gamma})\in\Pers_*({\Ucal}')$
with a representative  $(z,A)$.
We have the following cases:
\begin{itemize}
    \item $\sG_0\not\in z$:
    The chain $A$ may contain some cofaces of $\sG_0$ 
    but the number of cofaces of $\sG_0$ in $A$ must be even.
    If we consider $A$ as a  chain in $\tilde{\Ucal}$,
    we have: (i) $A$ contains an even number of cofaces of $\sG_1$ 
    and an even number of cofaces of $\sG_2$;
    (ii) $A$ contains an odd number of cofaces of $\sG_1$ 
    and an odd number of cofaces of $\sG_2$.
    In the former case,  the boundary of $A$ still equals $z$ in $\tilde{\Ucal}$
    so that 
    $(z,A)$ is still a  representative for 
    the corresponding pair in
    $\Pers_*(\tilde{\Ucal})$.
    In the latter case, $\partial(A)=z+\sG_1+\sG_2$ in $\tilde{\Ucal}$.
    Letting $A'=A+\extsimp$, we have $\partial(A')=z$ in $\tilde{\Ucal}$.
    Moreover, $A'$ is still  created by $\gamma$.
    To see this,
    first notice that $A$ contains
    a coface of $\sG_2$ in $\tilde{\Ucal}$.
    Then,
    as the creator of $A$, the cell $\gamma$ must be added
    after a coface of $\sG_2$
    in $\tilde{\Ucal}$,
    which means that $\gamma$ is added after $\sG_2$
    in $\tilde{\Ucal}$.
    So we have $\gamma$ is added after 
    ${\extsimp}$ in $\tilde{\Ucal}$
    and so $A'$ is still  created by $\gamma$.
    Therefore, $(z,A')$ is a representative for 
    the corresponding pair in
    $\Pers_*(\tilde{\Ucal})$.
    
    \item $\sG_0\in z$:
    If we consider $A$ as a  chain in $\tilde{\Ucal}$,
    we have: (i) $A$ contains an odd number of cofaces of $\sG_1$ 
    and an even number of cofaces of $\sG_2$;
    (ii) $A$ contains an even number of cofaces of $\sG_1$ 
    and an odd number of cofaces of $\sG_2$.
    Let $z'$ be the boundary of $A$ where we consider both $z'$ and $A$ as
    chains in $\tilde{\Ucal}$.
    In case~(i), 
    the only difference of $z$ and $z'$ is that $\sG_0$ in $z$
    is replaced by $\sG_1$ in $z'$.
    Since the addition of $\sG_0$ in $\Ucal'$
    and the addition of $\sG_1$ in $\tilde{\Ucal}$
    are at the same position,
    the creator of $z$ and $z'$ are the same if 
    $\eta\neq\sG_0$
    and the creator of $z'$ becomes ${\sG_1}$
    if $\eta=\sG_0$.
    Either way,
    $(z',A)$ is a  representative for 
    the corresponding pair in
    $\Pers_*(\tilde{\Ucal})$.
    In case (ii), 
    we have that $\sG_0$ in $z$
    is replaced by $\sG_2$ in $z'$.
    Letting $A'=A+\extsimp$,
    we have $\partial(A')=\partial(A)+\partial(\extsimp)=z'+\sG_1+\sG_2=z'\setminus\Set{\sG_2}\union\Set{\sG_1}$.
    Letting $z'':=z'\setminus\Set{\sG_2}\union\Set{\sG_1}$,
    we have that 
    the only difference of $z$ and $z''$ is that $\sG_0$ in $z$
    is replaced by $\sG_1$ in $z''$.
    With arguments similar as for the  case of $\sG_0\not\in z$,
    we have that
    $(z'',A')$ is a representative for 
    the corresponding pair in
    $\Pers_*(\tilde{\Ucal})$.
\end{itemize}

Re-writing representatives for open-closed $\Dim$-pairs in 
$\Pers_*(\Ucal')$ is symmetric to the re-writing for closed-open $\Dim$-pairs
described above.
Now consider a  closed-closed $\Dim$-pair $(\add{\eta},\del{\gamma})\in\Pers_*(\Ucal')$
with a representative tuple $(z_1,A,z_2)$,
where $z_1+z_2=\partial(A)$.
We have the following cases:
\begin{itemize}
    \item $\sG_0\not\in z_1+z_2$:
    We have
    the following sub-cases:
    \begin{itemize}
        \item $\sG_0$ is in neither of $z_1$ and $z_2$:
        If we consider $A$ as a chain in $\tilde{\Ucal}$,
    then $\partial(A)$  equals $z_1+z_2$ 
    or $z_1+z_2+\sG_1+\sG_2$ (see the representative re-writing for 
    closed-open $\Dim$-pairs). 
    If $\partial(A)=z_1+z_2$ in $\tilde{\Ucal}$, let $A'=A$;
    otherwise, let $A'=A+\extsimp$.
    It can  be verified that
    $(z_1,A',z_2)$ is a valid
    representative for the corresponding pair in
    $\Pers_*(\tilde{\Ucal})$.

        \item
    $\sG_0$ is in both $z_1$ and $z_2$:
    Let $z'_1$ (resp.\ $z'_2$) be a chain derived from $z_1$ (resp.\ $z_2$)
    by replacing $\sG_0$ with $\sG_1$ (resp.\ $\sG_2$).
    Notice that $z'_1+z'_2=z_1+z_2+\sG_1+\sG_2$.
    Moreover,
    if we consider $A$ as a chain in $\tilde{\Ucal}$,
    then $\partial(A)$  equals $z_1+z_2$ 
    or $z_1+z_2+\sG_1+\sG_2$.
    If $\partial(A)=z_1+z_2$ in $\tilde{\Ucal}$, let $A'=A+\extsimp$;
    otherwise, let $A'=A$.
    It can  be verified that
    $(z'_1,A',z'_2)$ is a valid
    representative for the corresponding pair in
    $\Pers_*(\tilde{\Ucal})$.
        \end{itemize}
    
    \item $\sG_0\in z_1+z_2$:
    We have
    the following sub-cases:
    \begin{itemize}
    \item $\sG_0\in z_1$, $\sG_0\not\in z_2$: 
    Based on the representative re-writing for 
    closed-open $\Dim$-pairs,
    the boundary of $A$ in $\tilde{\Ucal}$
    is derived from
    its boundary in $\Ucal'$ with $\sG_0$  replaced by $\sG_1$
    or $\sG_2$.
    If $\sG_0$ is replaced by $\sG_1$ in the boundary of $A$
    from $\Ucal'$ to $\tilde{\Ucal}$, 
    let $A'=A$;
    otherwise, let $A'=A+\extsimp$.
    Notice that 
    the boundary of $A'$ in $\tilde{\Ucal}$ always equals
    the boundary of $A$ in ${\Ucal}'$
    with $\sG_0$  replaced by $\sG_1$.
    Moreover, let $z'_1$ be a chain derived from $z_1$ 
    by replacing $\sG_0$ with $\sG_1$.
    Then, $(z'_1,A',z_2)$ is a valid
    representative for the corresponding pair in
    $\Pers_*(\tilde{\Ucal})$.
    
    \item $\sG_0\not\in z_1$, $\sG_0\in z_2$: 
    The re-writing is symmetric.
    \end{itemize}
\end{itemize}

\paragraph{$(\Dim+1)$-pairs.}
First consider a closed-open $(\Dim+1)$-pair in $\Pers_*({\Ucal}')$
with a representative  $(z,A)$,
where $z$ is a $(\Dim+1)$-cycle and
$A$ is a $(\Dim+2)$-chain.
We have that the boundary of $A$ is still $z$ in $\tilde{\Ucal}$
and so $z$ is still a cycle in $\tilde{\Ucal}$.
Therefore, $(z,A)$ is also a representative for the
corresponding pair in $\Pers_*(\tilde{\Ucal})$.
For open-closed $(\Dim+1)$-pairs in $\Pers_*({\Ucal}')$,
we also do not need to change the representatives.
For a closed-closed $(\Dim+1)$-pair  $(\add{\eta},\del{\gamma})\in\Pers_*({\Ucal}')$
with a representative  $(z_1,A,z_2)$,
we have that $z_1+z_2$ is still a cycle in $\tilde{\Ucal}$.
However, since $z_1$ (resp. $z_2$) may contain an even number of cofaces of $\sG_1$
and an even number of cofaces of $\sG_2$ in $\tilde{\Ucal}$,
the boundary of $z_1$ (resp. $z_2$) may become $\sG_1+\sG_2$.
If $z_1$ and $z_2$ are still cycles in  $\tilde{\Ucal}$,
then $(z_1,A,z_2)$ is also a representative for the
corresponding pair in $\Pers_*(\tilde{\Ucal})$;
otherwise, it can be verified that $(z_1+\extsimp,A,z_2+\extsimp)$ is  a representative for the
corresponding pair in $\Pers_*(\tilde{\Ucal})$.

\section{Inward contraction}
\label{sec:fzz-in-contra}

Consider the following inward contraction from $\Fcal$ to $\Fcal'$:
\begin{equation*}
\begin{tikzpicture}[baseline=(current  bounding  box.center)]
\node (a) at (0,0) {$\Fcal: K_0 \leftrightarrow
\cdots\leftrightarrow K_{i-2}
\leftrightarrow 
K_{i-1}\inctosp{\sG} 
K_i 
\bakinctosp{\sG} K_{i+1}
\leftrightarrow K_{i+2}\leftrightarrow
\cdots \leftrightarrow K_\filtcnt$}; 
\node (b) at 
(0,-1.02){$\Fcal': K_0 \leftrightarrow
\cdots\leftrightarrow K_{i-2}
\leftrightarrow 
K'_{i}
\leftrightarrow K_{i+2}\leftrightarrow
\cdots \leftrightarrow K_\filtcnt$};
\path[->] ([xshift=7pt]a.south) edge[double distance=2pt,arrows=-{Classical TikZ Rightarrow[length=4pt]}]  ([xshift=7pt]b.north);
\end{tikzpicture}
\end{equation*}
Let $\Ucal$ and $\Ucal'$ be the up-down filtrations corresponding to
$\Fcal$ and $\Fcal'$ respectively.
Moreover, let $\hat{\sG}$ be the $\DG$-cell corresponding to 
the addition of the simplex $\sG$ to $K_{i-1}$ in $\Fcal$.
We first perform forward and backward switches on $\Ucal$ to 
place the additions and deletions of $\hat{\sG}$ in the center,
obtaining the filtration 
\begin{equation*}
\Ucal^+:
\bullet
\inctosp{\usimp_0}
\bullet
\inctosp{\usimp_{1}}
\bullet
\cdots
\bullet\inctosp{\usimp_{\simpcnt-2}} 
\bullet
\inctosp{\hat{\sG}} 
\bullet
\bakinctosp{\hat{\sG}}
\bullet
\bakinctosp{\usimp_{\simpcnt-1}}
\bullet
\bakinctosp{\usimp_{\simpcnt}}
\bullet
\cdots
\bullet
\bakinctosp{\usimp_{\filtcnt-3}}
\bullet
\end{equation*}
as in Equation~(\ref{eqn:U-plus})
which we also call $\Ucal^+$.
Hence, we can obtain $\Pers_*(\Ucal^+)$ and its representatives
in $O(\filtcnt^2)$ time.
To 
perform the update 
from $\Ucal^+$ to $\Ucal'$
(an inward contraction on up-down filtrations),
we observe an
interval mapping behavior
(i.e., the alternative re-linking of `unsettled' pairs)
different from the one proposed in~\cite{maria2014zigzag,maria2016computing}
for inward expansion (see, e.g., 
Theorem 2.3
in~\cite{maria2014zigzag}).

The rest of the update has different processes based on whether 
$\add{\hat{\sG}}$ is positive or negative in $\Ucal^+$.

\subsection{${\scriptstyle \searrow}\hat{\sG}$ is positive in $\Ucal^+$}

If $\add{\hat{\sG}}$ is positive in $\Ucal^+$,
then $\add{\hat{\sG}}$ must be paired with $\del{\hat{\sG}}$
in $\Pers_*(\Ucal^+)$
and to obtain $\Pers_*(\Ucal')$
one only needs to delete the pair $(\add{\hat{\sG}},\del{\hat{\sG}})$
from $\Pers_*(\Ucal^+)$; see~\cite[Theorem 2.2]{maria2014zigzag}.
We then show how to obtain  representatives for pairs in $\Pers_*(\Ucal')$.
Ignoring the pair $(\add{\hat{\sG}},\del{\hat{\sG}})$,
call those pairs in $\Pers_*(\Ucal^+)$ 
whose chains in the representatives contain $\hat{\sG}$
as \emph{unsettled}.
 Accordingly,
call those 
pairs in $\Pers_*(\Ucal^+)$ 
whose chains in the representatives do not contain $\hat{\sG}$
as \emph{settled}.
Notice that
representatives for settled pairs stay
the same 
when we consider the pairs to be from $\Pers_*(\Ucal')$.
We also have that only closed-closed pairs in $\Pers_*(\Ucal^+)$  
can be unsettled
because the addition and deletion of $\hat{\sG}$ are at the center of $\Ucal^+$
(see Definition~\ref{dfn:ud-rep}).
Let $\bar{z}$ be a cycle containing $\hat{\sG}$ which is
taken from the representative for $(\add{\hat{\sG}},\del{\hat{\sG}})\in\Pers_*(\Ucal^+)$.
For an unsettled pair $(\add{\eta},\del{\gamma})\in\Pers_*(\Ucal^+)$
with a representative $(z,A,z')$,
we have  $\hat{\sG}\in A$.
Then, $(z,A+\bar{z},z')$
is a representative for 
$(\add{\eta},\del{\gamma})\in\Pers_*(\Ucal')$,
where  $z+z'=\partial(A)=\partial(A+\bar{z})$
and $\hat{\sG}\not\in A+\bar{z}$.

\subsection{${\scriptstyle \searrow}\hat{\sG}$ is negative in $\Ucal^+$}

The update for this case has two steps:
We first preprocess those unsettled
pairs (definition similar as above) where
$\hat{\sG}$ in the representatives
can be easily eliminated
so that the pairs become settled.
After this, the remaining unsettled
pairs satisfy a certain criteria and we use
the alternative re-linking to produce new pairs
for $\Ucal'$.

\paragraph{Step I: Preprocessing.}
Since $\add{\hat{\sG}}$ is negative in $\Ucal^+$,
 $\del{\hat{\sG}}$ must be positive in $\Ucal^+$
(if adding $\hat{\sG}$ decreases the Betti number,
then deleting $\hat{\sG}$ must increase the Betti number).
Let $(\add{\tau},\add{\hat{\sG}})$ and $(\del{\hat{\sG}},\del{\tau'})$
be pairs in $\Pers_*(\Ucal^+)$ containing 
$\add{\hat{\sG}}$ and $\del{\hat{\sG}}$ respectively.
Also suppose that $(\add{\tau},\add{\hat{\sG}})$ 
has the representative $(z_*,A_*)$ and
$(\del{\hat{\sG}},\del{\tau'})$ has the representative $(A_\circ,z_\circ)$.
Moreover, let $\LG$ be the set of unsettled pairs 
in $\Pers_*(\Ucal^+)$.
Notice that
we do not include 
$(\add{\tau},\add{\hat{\sG}})$ and $(\del{\hat{\sG}},\del{\tau'})$
into $\LG$
and so pairs in $\LG$ are all closed-closed.

\begin{definition}
Define a partial order `$\ccleq$' on closed-closed pairs
in $\Pers_*(\Ucal^+)$ s.t.\ 
$(\add{\eta},\del{\gamma})\ccleq(\add{\eta'},\del{\gamma'})$ 
if and only if  $\eta$ is added after $\eta'$
and $\gamma$ is deleted before $\gamma'$ in $\Ucal^+$
(i.e., the persistence interval generated by $(\add{\eta},\del{\gamma})$
is contained in the persistence interval 
generated by $(\add{\eta'},\del{\gamma'})$).
We also say that two pairs
$(\add{\eta},\del{\gamma})$ and $(\add{\eta'},\del{\gamma'})$ 
are \defemph{comparable} if $(\add{\eta},\del{\gamma})\ccleq(\add{\eta'},\del{\gamma'})$ 
or $(\add{\eta'},\del{\gamma'})\ccleq(\add{\eta},\del{\gamma})$.
\end{definition}

We then make certain pairs in $\LG$ settled and delete these pairs from $\LG$
so that $\LG$ satisfies the following:
(i) no two pairs in $\LG$ are comparable;
(ii) for each pair $(\add{\eta},\del{\gamma})$ in $\LG$,
the cell $\eta$ is added before $\tau$
and the cell $\gamma$ is deleted after $\tau'$ in $\Ucal^+$.

First do the following:
\begin{enumerate}
    \item While there is a pair $(\add{\eta},\del{\gamma})$ in $\LG$
    s.t.\ $\eta$ is added after $\tau$ in $\Ucal^+$:
    Let $(z,A,z')$ be the representative maintained for $(\add{\eta},\del{\gamma})$.
    Assign a new representative $(z+z_*,A+A_*,z')$
    to $(\add{\eta},\del{\gamma})$. The new representative  
    is valid because:
    (i) $z+z_*+z'=\partial(A)+\partial(A_*)=\partial(A+A_*)$;
    (ii) $z+z_*$ is also created by $\eta$ because 
    $\eta$ is added after the creator $\tau$ of $z_*$.
    Since $\hat{\sG}\not\in A+A_*$,
    the pair $(\add{\eta},\del{\gamma})$ becomes settled,
    and we  delete $(\add{\eta},\del{\gamma})$ from $\LG$.
    Figure~\ref{fig:in-contr-preproc1} illustrates the change on the representative
    for $(\add{\eta},\del{\gamma})$.

  \begin{figure}[!tb]
  \centering
  \includegraphics[width=0.69\linewidth]{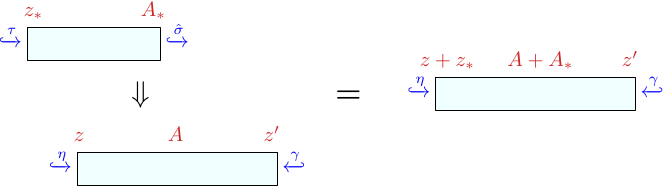}
  \caption{Representative for 
  $({\scriptstyle \searrow}\tau,{\scriptstyle \searrow}\hat{\sG})$ 
  can be considered
    as being `summed to' the representative for
    $({\scriptstyle \searrow}{\eta},{\scriptstyle \nwarrow}{\gamma})$ 
    to obtain a new representative
    for $({\scriptstyle \searrow}{\eta},{\scriptstyle \nwarrow}{\gamma})$.
    $({\scriptstyle \searrow}{\eta},{\scriptstyle \nwarrow}{\gamma})$ then becomes settled.}
  \label{fig:in-contr-preproc1}
  \end{figure}

  \item (This step is symmetric to the previous step.)
  While there is a pair $(\add{\eta},\del{\gamma})$ in $\LG$
    s.t.\ $\gamma$ is deleted before $\tau'$ in $\Ucal^+$:
    Let $(z,A,z')$ be the representative maintained for $(\add{\eta},\del{\gamma})$.
    Assign a new representative $(z,A+A_\circ,z'+z_\circ)$
    to $(\add{\eta},\del{\gamma})$.
    Delete $(\add{\eta},\del{\gamma})$ from $\LG$,
    which becomes settled. 
    Figure~\ref{fig:in-contr-preproc2} illustrates the change on the representative.
    \begin{figure}[!tb]
  \centering
  \includegraphics[width=0.69\linewidth]{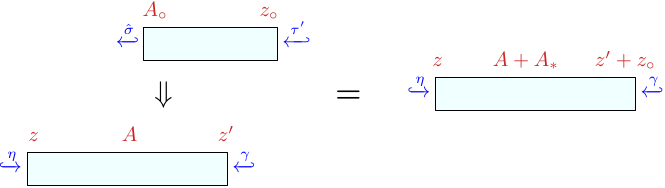}
  \caption{The representative for 
  $({\scriptstyle \nwarrow}\hat{\sG},{\scriptstyle \nwarrow}{\tau'})$ 
  is `summed to' the representative for
    $({\scriptstyle \searrow}{\eta},{\scriptstyle \nwarrow}{\gamma})$ 
    to obtain a new representative
    for $({\scriptstyle \searrow}{\eta},{\scriptstyle \nwarrow}{\gamma})$.
    $({\scriptstyle \searrow}{\eta},{\scriptstyle \nwarrow}{\gamma})$ then becomes settled.}
  \label{fig:in-contr-preproc2}
  \end{figure}
\end{enumerate}

Then do the following:

\begin{itemize}
    \item While there are two pairs in $\LG$ s.t.\ 
    $(\add{\eta},\del{\gamma})\ccleq(\add{\eta'},\del{\gamma'})$:
    Let $(z_l,A,z_r)$ and $(z'_l,A',z'_r)$ be the representatives
    for $(\add{\eta},\del{\gamma})$ and $(\add{\eta'},\del{\gamma'})$
    respectively. Sum $(z'_l,A',z'_r)$ to $(z_l,A,z_r)$ to
    form a new representative $(z_l+z'_l,A+A',z_r+z'_r)$
    for $(\add{\eta},\del{\gamma})$.
    Delete $(\add{\eta},\del{\gamma})$ from $\LG$,
    which becomes settled.
    See Figure~\ref{fig:in-contr-comparable} for an illustration.
\end{itemize}
\begin{figure}[!tb]
  \centering
  \includegraphics[width=0.8\linewidth]{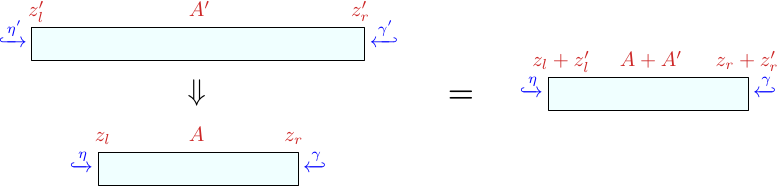}
  \caption{The representative for 
  $({\scriptstyle \searrow}\eta',{\scriptstyle \nwarrow}\gamma')$ 
  is `summed to' the representative for
    $({\scriptstyle \searrow}{\eta},{\scriptstyle \nwarrow}{\gamma})$ 
    to obtain a new representative
    for $({\scriptstyle \searrow}{\eta},{\scriptstyle \nwarrow}{\gamma})$.
    $({\scriptstyle \searrow}{\eta},{\scriptstyle \nwarrow}{\gamma})$ then becomes settled.}
  \label{fig:in-contr-comparable}
\end{figure}

Finally, we let each settled pair in $\Pers_*(\Ucal^+)$
automatically become a pair in $\Pers_*(\Ucal')$
with the same representative.

\paragraph{Step II: Alternatively re-linking unsettled pairs.}
After the preprocessing step, 
if $\LG=\emptyset$,
then $(\add{\tau},\del{\tau'})$ forms a pair in $\Pers_*(\Ucal')$
with a representative $(z_*,A_*+A_\circ,z_\circ)$
and we have finished the update.
If $\LG\neq\emptyset$,
we order the pairs in 
$\LG$ as 
$\Set{(\add{\eta_j},\del{\gamma_j})\mid j=1,2,\ldots,\ell}$
s.t.\ $\eta_j$ is added before $\eta_{j+1}$ for each $j$.
Since pairs in $\LG$ are not comparable,
we have that $\gamma_j$ is deleted before $\gamma_{j+1}$
for each $j$ because it  would then be true that 
$(\add{\eta_{j+1}},\del{\gamma_{j+1}})\ccleq(\add{\eta_j},\del{\gamma_j})$
if $\gamma_j$ is deleted after $\gamma_{j+1}$
for a $j$.
See Figure~\ref{fig:in-contr-relink}.
Additionally, let each $(\add{\eta_j},\del{\gamma_j})$ have a representative
$(z_j,A_j,z'_j)$.

Now, alternatively re-link the additions and deletions in the unsettled pairs
as follows (see Figure~\ref{fig:in-contr-relink}):

\begin{enumerate}
\item Form a pair $(\add{\eta_1},\del{\tau'})\in\Pers_*(\Ucal')$ with a representative
$(z_1,A_1+A_\circ,z'_1+z_\circ)$,
where $\hat{\sG}\not\in A_1+A_\circ$.

\item Form a pair $(\add{\tau},\del{\gamma_\ell})\in\Pers_*(\Ucal')$ with a representative
$(z_\ell+z_*,A_\ell+A_*,z'_\ell)$,
where $\hat{\sG}\not\in A_\ell+A_*$.

\item For each $j=1,\ldots,\ell-1$:
Form a pair $(\add{\eta_{j+1}},\del{\gamma_{j}})\in\Pers_*(\Ucal')$ with a representative
$(z_j+z_{j+1},A_j+A_{j+1},z'_j+z'_{j+1})$,
where $\hat{\sG}\not\in A_j+A_{j+1}$.
\end{enumerate}
\begin{figure}[!tb]
  \centering
  \includegraphics[width=0.55\linewidth]{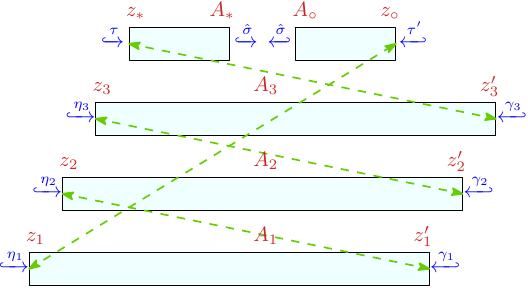}
  \caption{An example of $\ell=3$ where the dashed lines
  alternatively re-link the additions and deletions to form new 
  pairs in $\Pers_*(\Ucal')$.}
  \label{fig:in-contr-relink}
\end{figure}

We now conclude:
\begin{theorem}
Zigzag barcode for an inward contraction can
be updated in $O(\filtcnt^2)$ time.
\end{theorem}
\begin{proof}
The correctness of the updating follows from the correctness
of the representatives we set for $\Pers_*(\Ucal')$,
which can be verified from the operations presented above.
The costliest steps in the update are the chain summations, 
and there are $O(\filtcnt)$  of them each taking $O(m)$ time.
So the time complexity is $O(\filtcnt^2)$.
\end{proof}

\section{Conclusion} 
We have presented update algorithms 
for maintaining barcodes
over a changing zigzag filtration. Two
main questions come out of this research: 
(i) The representatives for the barcodes
are maintained only in the up-down filtrations.  
An open question that remains is whether this can help maintain the representatives for the original zigzag filtration in quadratic time over all operations. 
(ii) Are there other interesting applications of the update algorithms presented in this paper?  Their applications to computing vineyards for dynamic point clouds, 
multiparameter persistence, and dynamic network are mentioned in~\cite{dey2021updating,tymochko2020using}. 
We hope that there are other dynamic settings which can benefit from efficient zigzag updates.

\bibliographystyle{plainurl}
\bibliography{refs}

\appendix

\section{Details of the update in Section~\ref{sec:ez-up}}
\label{sec:ez-up-details}

\subsection{Forward/backward switch}
We only present the details for the update for forward switch
as  everything is symmetric for backward switch.
Consider the forward switch in Equation~(\ref{eqn:fwd-sw}).
Let $\Ud$ and $\Ud'$ be the corresponding up-down filtrations for $\Fcal$ and $\Fcal'$ respectively.
By the conversion in Section~\ref{sec:fzz-up},
there is also a forward switch 
(on the first half containing only additions)
from $\Ud$ to $\Ud'$
shown below,
where $\hat{\sG},\hat{\tG}$ are $\DG$-cells corresponding to $\sG,\tG$ respectively:
\begin{equation}\label{eqn:fwd-switch-ud}
\begin{tikzpicture}[baseline=(current  bounding  box.center)]
\tikzstyle{every node}=[minimum width=24em]
\node (a) at (0,0) {$\Ud:L_0\incto\cdots\incto
L_{j-1}\inctosp{\hat{\sG}}
L_{j}\inctosp{\hat{\tG}}
L_{j+1}
\incto
\cdots
\incto
L_{\simpcnt}
\bakincto
\cdots
\bakincto L_{\filtcnt}$}; 
\node (b) at (0,-0.6){$\Ud':L_0\incto\cdots\incto
L_{j-1}\inctosp{\hat{\tG}}
L'_{j}\inctosp{\hat{\sG}}
L_{j+1}
\incto
\cdots
\incto
L_{\simpcnt}
\bakincto
\cdots
\bakincto L_{\filtcnt}$};
\path[->] (a.0) edge [bend left=90,looseness=1.5,arrows={-latex},dashed] (b.0);
\end{tikzpicture}
\end{equation}

We can utilize the algorithm for the transposition diamond
in~\cite{maria2014zigzag} (Algorithm 6) for computing the 
update in Equation~(\ref{eqn:fwd-switch-ud}),
where a transposition diamond 
is indeed a forward switch.
Notice that Algorithm 6 in~\cite{maria2014zigzag} works on
a special type of zigzag filtration where the first part contains
additions or deletions
but the second part contains only deletions.
Since an up-down filtration satisfies the condition for
the special zigzag filtration in~\cite{maria2014zigzag},
Algorithm 6 in~\cite{maria2014zigzag} can be utilized in our setting.
The time complexity for performing the update is $O(\filtcnt)$~\cite{maria2014zigzag}.

\subsection{Inward expansion}
We only provide details for performing the update from
$\Ucal'$ to $\Ucal^+$.
Since an inward expansion takes place from $\Ucal'$ to $\Ucal^+$,
we notice that the interval mapping for an inward expansion is described
by theorems for the {injective} and {surjective diamond}
in~\cite{maria2014zigzag} (Theorem 2.2, 2.3).
An injective (resp.\ surjective) diamond is indeed an inward expansion
where the induced map $\Hm_*(K_{i-1})\to\Hm_*(K_i)$ in Equation~(\ref{eqn:in-contrac-expan})
is injective (resp.\ surjective).
We can then utilize the algorithm
for the {injective} and {surjective diamond}
in~\cite{maria2014zigzag}
(Algorithm 4 and 5) for computing the 
update from
$\Ucal'$ to $\Ucal^+$.
Notice that 
Algorithm 4 and 5 in~\cite{maria2014zigzag} have the same assumption
on an input zigzag filtration 
as the algorithm for the transposition diamond~\cite{maria2014zigzag}
mentioned above,
i.e.,
the first part of the zigzag filtration contains
additions or deletions
but the second part contains only deletions.
The time complexity for performing the update is $O(\filtcnt^2)$~\cite{maria2014zigzag}.

\subsection{Computing representatives for an initial up-down filtration}
The algorithm presented in~\cite{maria2014zigzag} 
can be considered as a process which iteratively 
maintains the barcode and representatives for a partial zigzag filtration 
derived from the input zigzag filtration.
In each iteration of the algorithm~\cite{maria2014zigzag},
the partial zigzag filtration grows and
eventually becomes the complete 
input filtration.
The algorithm maintains the barcode and representatives
by taking \emph{summations} of the intervals and their
representatives
in each iteration.
Notice that the algorithm achieves  
the claimed complexity (which is $O(\filtcnt^3)$
for an up-down filtration of length $\filtcnt$)
because it does not maintain the full representative
for an interval 
(which could contain $O(m')$ cycles for a general zigzag filtration
of length $m'$).
Instead, the algorithm maintains a representative cycle
\emph{at a single index} so that summing the representative cycles
of two intervals takes $O(m)$ time (because each cycle
contains $O(m)$ cells).
Since we are working on up-down filtrations in this paper,
we can modify the algorithm in~\cite{maria2014zigzag}
so  that whenever we sum two intervals in the algorithm,
we also sum the full representatives
(as in Definition~\ref{dfn:ud-rep}) for the two intervals.
Since a full representative in an up-down filtration
contains at most three cycles/chains
(see Definition~\ref{dfn:ud-rep}), 
the time complexity for computing 
representatives for an initial up-down filtration
is still $O(m^3)$.

\cancel{
\section{Potential applications of zigzag update algorithms}
\label{sec:zzup-app}

{\red delet this now?}

\paragraph{Dynamic point cloud.}
Consider a set of points $P$ moving with respect to time~\cite{edelsbrunner2012medusa,kim2020spatiotemporal}.
For each point pair in $P$, we can draw its \emph{distance-time curve}
revealing the variation of distance between the points w.r.t.\ time.
For example, Figure~\ref{fig:grid_curve} draws the curves for a simple
$P$ with three points,
where $e_1,e_2$ and $e_3$ denote edges formed by the three point pairs.
Consider the Vietoris-Rips complex of $P$ with $\dG$ as the distance threshold.
Since distances of the point pairs may become greater or less than
$\dG$ at different time,
edges formed by these pairs
are added to or deleted from the Rips complex accordingly.
This forms a zigzag filtration of Rips complexes, which we denote as $\Rcal^\dG$.
Letting $\dG$ vary from $0$ to $\infty$, and taking the persistence diagram (PD)
of $\Rcal^\dG$, we obtain a vineyard~\cite{cohen2006vines} as a descriptor
for the dynamic point cloud.
We note that $\Rcal^\dG$ changes only
at the \emph{critical} points of the distance-time curves,
which are local minima/maxima and intersections
(as illustrated by the dots in Figure~\ref{fig:grid_curve}).
To compute the vineyard, one only needs to compute the PD
of each $\Rcal^\dG$ where $\dG$ is in between
distance values of two critical points.
For example, $\Set{\dG_i}_i$ are the distance values for the critical points
in Figure~\ref{fig:grid_curve},
and $\Set{d_i}_i$ are the values in between.
Figure~\ref{fig:grid_module} lists the zigzag filtration $\Rcal^{d_i}$
for each $d_i$, where each horizontal arrow is either an equality, addition of an edge,
or deletion of an edge.
Each transition from $\Rcal^{d_i}$ to $\Rcal^{d_{i+1}}$ can be realized by
a sequence of atomic operations described in this paper,
which provides natural associations for the PDs~\cite{cohen2006vines}.
For example, starting from the top and going down, 
one needs to perform forward/backward/outward switches,
inward contractions, and outward expansions
(defined in Section~\ref{sec:update-oper}).
One could also start from the bottom and go up,
which requires the reverse operations.
In Section~\ref{sec:dpc},
we provide details on how the zigzag filtrations are built
for a dynamic point cloud and how the atomic operations can 
be used to realize the transitions.

\begin{figure}[p]
  \centering
  \captionsetup{justification=centering}
  \captionsetup[subfigure]{justification=centering}

  \begin{subfigure}[t]{\textwidth}
  \centering
  \includegraphics[width=0.65\linewidth]{fig/grid_curve}
  \caption{Distance-time curves of the three point pairs.}
   \label{fig:grid_curve}
  \end{subfigure}
  
  \vspace{2em}
   
  \captionsetup[subfigure]{justification=justified}

  \begin{subfigure}[t]{\textwidth}
  \centering
  \includegraphics[width=0.65\linewidth]{fig/grid_module}
  \caption{Zigzag filtration $\Rcal^{d_i}$ for each $d_i$ is listed horizontally, while vertically 
  each Rips complex is included into the one on the above.}
   \label{fig:grid_module}
  \end{subfigure}

  \caption{An example of a dynamic point cloud with three points.}
  \label{fig:dpc-grid}
\end{figure}

\paragraph{Levelset zigzag for time-varying function.}
It is known that the level sets
of a function give rise to a special type of zigzag filtrations
called levelset zigzag filtrations~\cite{carlsson2009zigzag-realvalue}, 
which are known to capture more information than the non-zigzag sublevel-set filtrations. Thus, even for a time-varying function, computing
the vineyard for a levelset zigzag filtration may capture more information than
the one by non-zigzag filtrations. 

\paragraph{Other potential applications.}
We also hope that our algorithms
for maintaining the representatives 
may be of independent interest.
For example, an efficient maintenance of these representatives provided
an efficient algorithm for computing zigzag persistence on graphs~\cite{dey2021graph} 
and also explained why a persistence algorithm proposed by Agarwal et al.~\cite{agarwal2006extreme} 
for elevation
functions works. Hilbert (dimension) function or rank 
function are among
some of the basic features for a
multiparameter persistence module.
One may use zigzag updates to compute these functions more efficiently
as Figure~\ref{fig:use2d}a suggests.
Thinking forward, we see a potential use of our algorithms
for maintaining representatives to compute generalized rank invariants~\cite{KM20,Patel}
for 2-parameter persistence modules. This may help
compute different homological structures as 
advocated recently~\cite{DKM21}; see Figure~\ref{fig:use2d}b.

\begin{figure}[htbp]
    \centering
    \begin{tabular}{cc}
    \includegraphics[height=4cm]{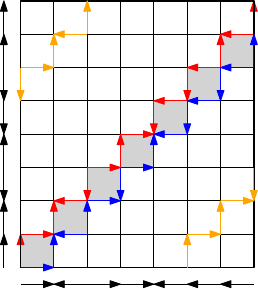}&
    \includegraphics[height=4cm]{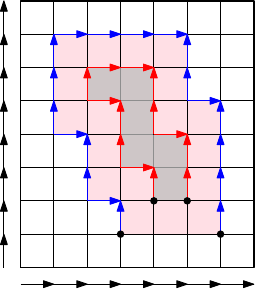}\\
    (a) & (b)
    \end{tabular}
        \caption{(a) Computing dimensions or rank function on a persistence module 
        with support over a 2D zigzag grid (poset) can be more efficiently computed by considering zigzag persistence on an initial zigzag filtration (indicated by red path) and then updating it with switches, which gives
        other zigzag paths (indicated by blue and golden paths). Assuming $t$ points
        in the grid, this will take $O(t^3)$ time with the updates instead of $O(t^{\omega+2})$
        with bruteforce zigzag persistence computation on every path. (b) Recently, it is
        shown that the generalized rank of an interval in a 2-parameter module can be derived
        from the zigzag persistence on the boundary as shown with red and blue paths for the
        grey and pink intervals respectively~\cite{DKM21}. We can leverage our update algorithms to compute
        the zigzag persistence over these two paths and multiple boundaries in general.}
    \label{fig:use2d}
\end{figure}

\subsection{Details on dynamic point clouds}
\label{sec:dpc}

We first define the following:

\begin{definition}
Throughout the section,
let $\dpc=(\pset,\posf_0,\posf_1,\ldots,\posf_\dpccnt)$
denote a dynamic point cloud in which: (i) $\pset$ is a set of points;
(ii) each map $\posf_i:\pset\to\Real^\Dim$ specifies the positions
of points in $\pset$ at \emph{time} $i$.
\end{definition}

Note that while 1-dimensional persistence~\cite{edelsbrunner2000topological}
with Rips filtration
serves as an effective descriptor for a fixed point cloud, 
it cannot naturally characterize a dynamic point cloud as defined above~\cite{kim2020spatiotemporal}.
In view of this,
we build vines and vineyards~\cite{cohen2006vines}
as descriptors for $\dpc$ using zigzag persistence.
We first let the time in $\dpc$ range {continuously}
in $[0,\dpccnt]$, i.e., the position of each point in $\pset$
during time $[0,\dpccnt]$ is linearly interpolated
based on 
the discrete samples given in $\dpc$.
For each $t\in[0,\dpccnt]$, let $\pset_t$ denote the point cloud
which is the point set $\pset$ with positions at time $t$.
Also, for a $\dG\geq 0$, let $R^\dG_t$ denote the Rips complex of $\pset_t$
with distance $\dG$.

Now fix a $\dG\geq 0$, and consider the continuous sequence 
$\Rcal^\dG:=\Set{R^\dG_t}_{t\in[0,\dpccnt]}$.
We claim that $\Rcal^\dG$ is encoded by a zigzag filtration,
and hence admits a barcode (persistence diagram)
as descriptor.
To see this, we note that each $R^\dG_t$ in $\Rcal^\dG$ is completely determined
by the vertex pairs in $\pset$ with distances no greater than $\dG$ at time $t$. 
Let $\pi$ be a vertex pair whose distance varies with time as illustrated
by the red curve in Figure~\ref{fig:dis_plot},
where the horizontal axis denotes time and the vertical axis denotes distance.
For the $\dG$ in Figure~\ref{fig:dis_plot},
the edge formed by $\pi$ is in $R^\dG_t$ when $t$ falls in the intervals
$[0,t_1]$, $[t_2,t_3]$, and $[t_4,t_5]$.
Also, in Figure~\ref{fig:pair_intervals}, for three vertex pairs $\pi_1,\pi_2,\pi_3$,
we illustrate respectively the time intervals in which 
their distances are no greater than $\dG$.
With the time varying, 
the edges formed by the vertex pairs are added to
or deleted from the Rips complex. 
As illustrated in Figure~\ref{fig:pair_intervals},
this naturally defines a zigzag filtration which we denote as $\Fcal^\dG$.
For example, $R_{t_2}^\dG$ in Figure~\ref{fig:pair_intervals}
is defined by edges formed by $\pi_2$ and $\pi_3$, and
$R_{t_5}^\dG$ 
is defined by edges formed by $\pi_1$ and $\pi_3$.

\begin{figure}[t]
  \centering
  \captionsetup[subfigure]{justification=centering}

  \begin{subfigure}[t]{0.4\textwidth}
  \centering
  \includegraphics[width=\linewidth]{fig/dis_plot}
  \caption{}
  \label{fig:dis_plot}
  \end{subfigure}
  \hspace{1.5em}
  \begin{subfigure}[t]{0.5\textwidth}
  \centering
  \includegraphics[width=\linewidth]{fig/pair_intervals}
  \caption{}
  \label{fig:pair_intervals}
  \end{subfigure}

  \caption{(a) Distance-time curves for two vertex pairs.
  (b) Time intervals for three vertex pairs $\pi_1,\pi_2,\pi_3$ in which distance is $\leq\dG$
  and the corresponding zigzag filtration $\Fcal^\dG$.}
\end{figure}

We then consider the one-parameter family of 
persistence diagrams
$\Set{\Bcal^\dG}_{\dG\in[0,\infty]}$,
with $\Bcal^\dG$ 
being the persistence diagram of 
$\Rcal^\dG$,
which forms a vineyard~\cite{cohen2006vines}. 
Treating each $\Bcal^\dG$ as a multi-set
of points in $\Real^2$,
the vineyard 
$\Set{\Bcal^\dG}_{\dG\in[0,\infty]}$
contains vines tracking the movement of points in persistence diagrams 
w.r.t.\ $\dG$.
For computing the vineyard $\Set{\Bcal^\dG}_{\dG\in[0,\infty]}$, 
we utilize
the update operations 
and algorithms presented in this paper.
As in~\cite{cohen2006vines},
our atomic update operations help associate points 
for persistence diagrams in $\Set{\Bcal^\dG}_{\dG\in[0,\infty]}$
without ambiguity, 
which is otherwise unavoidable if attempting to associate directly.
Let $\bar{\dG}$ be the maximum distance of vertex pairs at all time
in $\dpc$. We start with $\Rcal^{\bar{\dG}}$.
Since $R^{\bar{\dG}}_t$ equals a contractible (high-dimensional) simplex
at any $t$,
$\Bcal^{\bar{\dG}}$ contains only a 0-th interval $[0,\dpccnt]$
whose representative sequence
is straightforward\footnote{In practice,
  one may only consider simplices up to a dimension to save time;
  $\Bcal^{\bar{\dG}}$ and the representatives 
  in this case
  can then be computed from a homology basis for the complex at a time $t$.}.
Now consider the distance-time curves of \emph{all} vertex pairs of $\pset$
(e.g., Figure~\ref{fig:dis_plot} illustrates curves of two pairs),
which indeed defines a \emph{dynamic metric space}~\cite{kim2020spatiotemporal}.
When decreasing the distance $\dG$, 
$\Fcal^\dG$ changes only
at the following types of points
in the plot of all distance-time curves (see Figure~\ref{fig:dpc_events}):
\begin{figure}[t]
  \centering
  \captionsetup[subfigure]{justification=centering}
  \includegraphics[width=\linewidth]{fig/dpc_events}
  \caption{The events that change the zigzag filtration of $\Rcal^\dG$ as $\dG$ varies. 
  Each (partial) distance-time curve corresponds to a vertex pair, and for some events,
  edges formed by the vertex pairs are also denoted.}
  \label{fig:dpc_events}
\end{figure}
\begin{description}
  \item[I.\ Increasing crossing]:
  In Figure~\ref{fig:dpc_events},
  $e_1$ is deleted first at $t_3$ and then $e_2$ is deleted at $t_4$ in $\Rcal^{\dG_1}$.
  In $\Rcal^{\dG_2}$, the deletions of $e_1,e_2$ are switched.
  The switch of edge deletions 
  in the zigzag filtrations
  is realized by a sequence
  of simplex-wise \emph{backward switches}.

  \item[II.\ Decreasing crossing]:
  This is symmetric to the increasing crossing
  where additions of two edges are switched. It is realized by a sequence
  of simplex-wise \emph{forward switches}.

  \item[III.\ Opposite crossing]:
  In Figure~\ref{fig:dpc_events},
  $e_1$ is added first at $t_2$ and then $e_2$ is deleted at $t_3$
  in $\Rcal^{\dG_1}$.
  In $\Rcal^{\dG_2}$, the addition of $e_1$ and the deletion of $e_2$ are switched.
  The simplex-wise version of $\Fcal^{\dG_1}$ contains the following part
  \[R^{\dG_1}_{t_1}\inctosp{\sG_1}\cdots\inctosp{\sG_q} R^{\dG_1}_{t_3}
  \bakinctosp{\tG_1}\cdots\bakinctosp{\tG_r} R^{\dG_1}_{t_5},\]
  where $t_1,t_3,t_5$ are as defined in Figure~\ref{fig:dpc_events}.
  To obtain $\Fcal^{\dG_2}$, we do the following for each $i=1,\ldots,r$:
  \begin{itemize}
      \item 
  If $\tG_i$ is not equal to any of $\sG_1,\ldots,\sG_q$, 
  then use \emph{outward switches} to make $\bakinctosp{\tG_i}$ appear
  immediately before the additions of $\sG_1,\ldots,\sG_q$.
  If $\tG_i$ is equal to a $\sG_j$, first use outward switches 
  to make $\bakinctosp{\tG_i}$ appear
  immediately after $\inctosp{\sG_j}$. 
  Then, apply the \emph{inward contraction} on $\inctosp{\sG_j}\bakinctosp{\tG_i}$.
  Note that $\sG_j$ ($=\tG_i$) which contains both $e_1,e_2$
  does not exist in any complex $R^{\dG_2}_t$
  for $t$ a time shown in Figure~\ref{fig:dpc_events}
  because $e_1,e_2$ do not both exist in these complexes.
  \end{itemize}

  \item[IV.\ Local minimum]:
  In Figure~\ref{fig:dpc_events},
  an edge $e$ corresponding to the black curve 
  is added at $t_1$ and then deleted at $t_2$ in $\Rcal^{\dG_1}$.
  In $\Rcal^{\dG_2}$, the addition and deletion of $e$ disappear.
  Correspondingly,
  simplices containing $e$ are added and then deleted in $\Fcal^{\dG_1}$, 
  but in $\Fcal^{\dG_2}$, the addition and deletion of the above mentioned simplices do not exist.
  Hence, we need to perform \emph{inward contractions}.
  Note that before this, we may need to perform forward or backward switches
  to properly order the additions and deletions.
  (For example, suppose that $\sG$ is the last simplex added 
  due to the addition of $e$.
  However, if $\sG$ is not the first simplex deleted 
  due to the deletion of $e$,
  we need to perform backward switches to make 
  this true
  so that we can perform an inward contraction on $\sG$.)

  \item[V.\ Local maximum]:
  In Figure~\ref{fig:dpc_events},
  an edge $e$ corresponding to the black curve 
  exists in any complex $R^{\dG_1}_t$ for $t$ a time shown in the figure.
  However, in $\Rcal^{\dG_2}$,
  $e$ is deleted at $t_1$ and then added at $t_2$.
  Accordingly, we need to perform \emph{outward expansions}
  on simplices which are deleted and then added.
\end{description}

All the above five types of points appear in Figure~\ref{fig:dis_plot} 
with the numbering of types labelled.}

\end{document}